\documentclass{llncs}
\usepackage[utf8]{inputenc}
\usepackage{amssymb,amsmath,stmaryrd}
\usepackage{mathptmx} 
\usepackage{array} 
\usepackage{blkarray} %
\usepackage{tikz} %
\usetikzlibrary{arrows,shapes,snakes,backgrounds}
\usepackage[colorlinks=true, citecolor=blue,linkcolor=blue]{hyperref}
\usepackage{algorithm}
\usepackage{fullpage}
\usepackage[noend]{algorithmic}
\usetikzlibrary{arrows,shapes,snakes,backgrounds}
\usepackage{paralist}
\makeatletter
\renewcommand*{\@fnsymbol}[1]{\ensuremath{\ifcase#1\or \star\or \dagger\or \ddagger\or
       \mathsection\or \mathparagraph\or \|\or **\or \dagger\dagger
       \or \ddagger\ddagger \else\@ctrerr\fi}}
\makeatother
\DeclareSymbolFont{cmmathcal}{OMS}{cmsy}{m}{n}
\DeclareSymbolFontAlphabet{\mathcal}{cmmathcal}
\newcommand{\state}{\ensuremath{s} }
\newcommand{\states}{\ensuremath{S} }
\newcommand{\statesOne}{\ensuremath{S_{1}} }
\newcommand{\statesTwo}{\ensuremath{S_{2}} }
\newcommand{\initState}{\ensuremath{s_{init}} }
\newcommand{\edges}{\ensuremath{E} }
\newcommand{\dimension}{\ensuremath{k} }
\newcommand{\weight}{\ensuremath{w} }
\newcommand{\largestW}{\ensuremath{W} }
\newcommand{\largestWenc}{\ensuremath{V} }
\newcommand{\game}{\ensuremath{G} }
\newcommand{\gameFull}{\ensuremath{\game = \left( \statesOne, \statesTwo, \initState, \edges, \dimension, \weight\right)} }
\newcommand{\parity}{\ensuremath{p} }
\newcommand{\parityFull}{\ensuremath{p\colon \states \rightarrow \nat} }
\newcommand{\gameParFull}{\ensuremath{G_{p} = \left( \statesOne, \statesTwo, \initState, \edges, \dimension, \weight, \parity\right)} }

\newcommand{\gamePar}{\ensuremath{G_{p}} }
\newcommand{\integ}{\ensuremath{\mathbb{Z}} }
\newcommand{\nat}{\ensuremath{\mathbb{N}} }
\newcommand{\rat}{\ensuremath{\mathbb{Q}} }
\newcommand{\reals}{\ensuremath{\mathbb{R}} }
\newcommand{\player}{\ensuremath{\mathcal{P}} }
\newcommand{\playerOne}{\ensuremath{\mathcal{P}_{1}} }
\newcommand{\playerTwo}{\ensuremath{\mathcal{P}_{2}} }
\newcommand{\play}{\ensuremath{\pi} }
\newcommand{\plays}{\ensuremath{\textsf{Plays}(\game)} }
\newcommand{\prefixes}{\ensuremath{\textsf{Prefs}(\game)} }
\newcommand{\prefixesArg}[1]{\ensuremath{\textsf{Prefs}_{#1}(\game)} }
\newcommand{\playsPar}{\ensuremath{\textsf{Plays}(\gamePar)} }

\newcommand{\prefixesArgPar}[1]{\ensuremath{\textsf{Prefs}_{#1}(\gamePar)} }
\newcommand{\prefix}{\ensuremath{\rho} }
\newcommand{\finplay}{\ensuremath{\eta} }
\newcommand{\last}{\ensuremath{\textsf{Last}} }
\newcommand{\first}{\ensuremath{\textsf{First}} }
\newcommand{\el}{\ensuremath{\textsf{EL}} }
\newcommand{\elFull}{\ensuremath{\textsf{EL}(\prefix) = \sum_{i = 0}^{i = n - 1} w(s_{i}, s_{i+1})} }
\newcommand{\mpay}{\ensuremath{\textsf{MP}} }
\newcommand{\mpayFull}{\ensuremath{\textsf{MP}(\play) = \liminf_{n \rightarrow \infty} \frac{1}{n} \el (\play(n))}}
\newcommand{\minpar}{\ensuremath{\textsf{Par}} }
\newcommand{\infPlay}{\ensuremath{\textsf{Inf}} } 
 
\newcommand{\minparFull}{\ensuremath{\textsf{Par}(\play) = \min \left\lbrace \parity(s) \;\vert\; s \in \infPlay(\play) \right\rbrace } }
\newcommand{\dist}{\ensuremath{\mathcal{D}} }
\newcommand{\strat}{\ensuremath{\lambda} }
\newcommand{\strats}{\ensuremath{\Lambda} }

\newcommand{\stratsPureMemoryless}{\ensuremath{\Lambda^{PM}} }

\newcommand{\stratsPureFinite}{\ensuremath{\Lambda^{PF}} }
\newcommand{\stratsRandomizedMemoryless}{\ensuremath{\Lambda^{RM}} }

\newcommand{\outcomePar}{\ensuremath{\textsf{Outcome}_{\gamePar}} }
\newcommand{\objective}{\ensuremath{\phi} }
\newcommand{\supp}{\ensuremath{\textsf{Supp}} }
\newcommand{\objEL}{\ensuremath{\textsf{PosEnergy}_{\gamePar}(v_{0})} }
\newcommand{\objMP}{\ensuremath{\textsf{MeanPayoff}_{\gamePar}(v)} }
\newcommand{\objMPnoPar}{\ensuremath{\textsf{MeanPayoff}_{\game}(v)} }
\newcommand{\objPar}{\ensuremath{\textsf{Parity}_{\gamePar}} }
\newcommand{\objELFull}{\ensuremath{\objEL = \left\lbrace \play \in \playsPar \;\vert\; \forall\, n \geq 0 : v_{0} + \el(\play(n)) \in \nat^{\dimension} \right\rbrace } }
\newcommand{\objMPFull}{\ensuremath{\objMP = \left\lbrace \play \in \playsPar \;\vert\; \mpay(\play) \geq v\right\rbrace} }

\newcommand{\objParFull}{\ensuremath{\objPar = \left\lbrace \play \in \playsPar \;\vert\; \minpar(\play) \text{ mod } 2 = 0\right\rbrace} }
\newcommand{\event}{\ensuremath{\mathcal{A}} }

\newcommand{\proba}{\ensuremath{\mathbb{P}} }
\newcommand{\expect}{\ensuremath{\mathbb{E}} }
\newcommand{\tree}{\ensuremath{T} }
\newcommand{\DAG}{\ensuremath{D} }

\newcommand{\treeStates}{\ensuremath{Q} }
\newcommand{\node}{\ensuremath{\varsigma} }
\newcommand{\otherNode}{\ensuremath{\vartheta} }
\newcommand{\nodeLess}{\ensuremath{\preceq} }
\newcommand{\nodeEqual}{\ensuremath{\simeq} }
\newcommand{\dagStates}{\ensuremath{Q} }

\newcommand{\treeEdges}{\ensuremath{R} }
\newcommand{\dagEdges}{\ensuremath{R} }

\newcommand{\treeFull}{\ensuremath{\tree = \left( \treeStates, \treeEdges\right)} }

\newcommand{\translation}{\ensuremath{\textsf{Tr}} }
\newcommand{\degree}{\ensuremath{\textsf{degree}} }
\newcommand{\lab}{\ensuremath{\Theta} }
\newcommand{\enAncestors}{\ensuremath{\textsf{EnAnc}} }
\newcommand{\ancestors}{\ensuremath{\textsf{Anc}} }
\newcommand{\oldestEnAnc}{\ensuremath{\textsf{oea}} }
\newcommand{\ancPath}{\ensuremath{\rightsquigarrow} }
\newcommand{\mergeOp}{\ensuremath{\textsf{merge}} }
\newcommand{\initCredit}{\ensuremath{v_{0}} }
\newcommand{\pebblePath}{\ensuremath{\zeta} }
\newcommand{\pebblePathA}{\ensuremath{\alpha} }
\newcommand{\pebblePathB}{\ensuremath{\beta} }

\newcommand{\pebblePathE}{\ensuremath{\xi} }
\newcommand{\pebblePathNew}{\ensuremath{\eta} }
\newcommand{\pebblePathNewA}{\ensuremath{\kappa} }
\newcommand{\pebbleUp}{\ensuremath{\circlearrowleft} }

\newcommand{\dagRoot}{\ensuremath{r} }
\newcommand{\nodeNC}{\ensuremath{\chi} }
\newcommand{\nodeMerge}{\ensuremath{\node_{\textsf{m}}} }

\newcommand{\gadgets}{\ensuremath{K} }
\newcommand{\depth}{\ensuremath{l} }
\newcommand{\width}{\ensuremath{L} }
\newcommand{\winningNodes}{\ensuremath{\textsc{Win}} }

\newcommand{\occ}{\ensuremath{\textsf{occ}} }
\newcommand{\occCirc}{\ensuremath{\occ_{\circ}} }
\newcommand{\freq}{\ensuremath{\textsf{freq}_{\infty}} }
\newcommand{\ergodic}{\ensuremath{\mathcal{M}_{e}} }

\newcommand{\chain}{\ensuremath{\mathcal{M}_{c}} }

\newcommand{\winEP}{\ensuremath{\textsf{Win}} }
\newcommand{\pr}{\ensuremath{\gamma} }
\newcommand{\ec}{\ensuremath{\mathcal{C}} }

\newcommand{\mpvalue}{\ensuremath{\mathsf{val}} }

\newcommand{\maxE}{{\mathbb{C}}}

\newcommand{\Cpre}{{\sf Cpre}}

\newcommand{\minElems}{{\sf Min}_{\preceq}}

\newcommand{\Attr}{{\sf Attr}}

\newcommand{\epSCT}{epSCT}
\newcommand{\branching}{d}
\newcommand{\priorities}{m}
\newcommand{\dagSeq}{\ensuremath{(D_{i})_{0 \leq i \leq n}} }

\newcommand{\algoName}{{\sf CpreFP}}

\floatstyle{plain}
\newfloat{myalgo}{tbhp}{mya}
{\begin{myalgo}[#1]
\centering
\begin{minipage}{#2}
\begin{algorithm}[H]}%
{\end{algorithm}
\end{minipage}
\end{myalgo}}

\title{Strategy Synthesis for Multi-Dimensional Quantitative Objectives}

\author{Krishnendu Chatterjee\inst{1}$^{,}$\thanks{Author supported by Austrian Science Fund (FWF) Grant No P 23499-N23, FWF NFN Grant No S11407 (RiSE), ERC Start Grant (279307: Graph Games), Microsoft faculty fellowship.} \and Mickael Randour\inst{2}$^{,}$\thanks{Author supported by F.R.S.-FNRS. fellowship.} \and Jean-Fran\c{c}ois Raskin\inst{3}$^{,}$\thanks{Author supported by ERC Starting Grant (279499: inVEST).}}

\institute{
IST Austria (Institute of Science and Technology Austria) \\
\and Computer Science Department, Universit\'e de Mons (UMONS), Belgium \\
\and D\'epartement d'Informatique, Universit\'e Libre de Bruxelles (U.L.B.), Belgium
}

\begin{document}

\maketitle


\begin{abstract}
Multi-dimensional mean-payoff and energy games provide the mathematical 
foundation for the quantitative study of reactive systems, and play a central role in the emerging quantitative theory of verification and synthesis. In this work, we study the strategy synthesis problem for games with such multi-dimensional objectives along with a parity condition, a canonical way to express $\omega$-regular conditions. While in general, the winning strategies in such games may require infinite memory, for synthesis the most relevant problem is the construction of a finite-memory winning strategy (if one exists). Our main contributions are as follows. First, we show a tight exponential bound (matching upper and lower bounds) on the memory required for finite-memory winning strategies in both multi-dimensional mean-payoff and energy games along with parity objectives. This significantly improves the triple exponential upper bound for multi energy games (without parity) that could be derived from results in literature for games on VASS (vector addition systems with states). Second, we present an optimal symbolic and incremental algorithm to compute a finite-memory winning strategy (if one exists) in such games. Finally, we give a complete characterization of when finite memory of strategies can be traded off for randomness. In particular, we show that for one-dimension mean-payoff parity games, randomized memoryless strategies are as powerful as their pure finite-memory counterparts.
\end{abstract}

\newcommand{\Set}[1]{\{ #1 \}}

\section{Introduction}\label{introduction}

Two-player games on graphs provide the mathematical 
foundation to study many important problems in computer science.
Game-theoretic formulations have especially proved useful for synthesis \cite{Church62,RamadgeWonham87,pnueli_POPL89}, verification~\cite{AHK02}, refinement~\cite{FairSimulation}, and compatibility 
checking \cite{InterfaceTheories} of reactive systems, as well as 
in analysis of emptiness of automata~\cite{Thomas97}.    

Games played on graphs are repeated games that proceed for an infinite 
number of rounds.
The \textit{state} space of the graph is partitioned into player 1 states
and player 2 states (player 2 is adversary to player 1). 
The game starts at an initial state, and if the current state is a 
player 1 (resp. player 2) state, then player 1 (resp. player 2) chooses an outgoing \textit{edge}. This choice is made according to a \textit{strategy} of the player: given the sequence of visited states, a \textit{pure} (resp. \textit{randomized}) strategy chooses an outgoing edge (resp. probability distribution over outgoing edges). This process of choosing edges is repeated forever, and gives rise to an outcome of the game, called a {\em play}, that consists of the 
infinite sequence of states that are visited. When randomized strategies are used, there is in general not a unique outcome, but a set of possible outcomes, as the choice of edges is stochastic rather than deterministic.

Traditionally, games on graphs have been studied with Boolean 
objectives such as reachability, liveness, $\omega$-regular conditions 
formalized as the canonical parity objectives, strong fairness 
objectives, etc~\cite{GH82,EJ88,EJ91,Zie98,Thomas97,WilkeBook}. 
While games with \emph{quantitative} objectives have been studied 
in the game theory literature~\cite{EM79,ZP95,Mar98}, their application
in synthesis and other problems in verification is quite recent.
The two classical quantitative objectives that are most relevant in 
verification and synthesis are the \emph{mean-payoff} and \emph{energy} 
objectives. 
In games on graphs with quantitative objectives, the game graph is equipped 
with a weight function that assigns integer-valued weights to every edge.
For mean-payoff objectives, the goal of player~1 is to ensure that the
long-run average of the weights is above a threshold.
For energy objectives, the goal of player~1 is to ensure that the sum of
the weights stays above~0 at all times.
In applications of verification and synthesis, the quantitative objectives
that typically arise are (i)~multi-dimensional quantitative objectives 
(i.e., conjunction of several quantitative objectives), e.g., to express 
properties like the average response time between a grant and a request is 
below a given threshold $\nu_1$, and the average number of unnecessary grants 
is below threshold $\nu_2$; and
(ii)~conjunction of quantitative objectives with a Boolean objective, such as a 
mean-payoff parity objective that can express properties like the average 
response time is below a threshold along with satisfying a liveness property.
In summary, the quantitative objectives can express properties related to 
resource requirements, performance, and robustness; 
multiple objectives can express the different, potentially dependent or 
conflicting objectives; and the Boolean objective specifies functional 
properties such as liveness or fairness.
The game theoretic framework of multi-dimensional quantitative games and games 
with conjunction of quantitative and Boolean objectives has recently been shown to 
have many applications in verification and synthesis, such as 
synthesizing systems with quality guarantee~\cite{bloem_CAV09},
synthesizing robust systems~\cite{BGHJ09},
performance aware synthesis of concurrent data structure~\cite{CCHRS11},
analyzing permissivity in games and synthesis~\cite{bouyer_ATVA11},
simulation between quantitative automata~\cite{CDH10}, 
generalizing Boolean simulation to quantitative simulation 
distance~\cite{CHR12}, etc.
Moreover, multi-dimensional energy games are equivalent to a decidable class of games on 
VASS (vector addition systems with states). This model is equivalent to games over multi-counter systems and Petri nets~\cite{brazdil_ICALP10}.

In literature, there are many recent works on the theoretical analysis 
of multi-dim\-en\-sional quantitative games, such as,
mean-payoff parity games~\cite{chatterjee_LICS05,bouyer_ATVA11}, energy-parity 
games~\cite{chatterjee_ICALP10}, multi-dimensional energy games~\cite{chatterjee_FSTTCS10}, and 
multi-dimensional mean-payoff games~\cite{chatterjee_FSTTCS10,VR11}.
Most of these works focus on establishing the computational complexity of the 
problem of deciding if player~1 \textit{has} a \emph{winning} strategy. 
From the perspective of synthesis and other related problems in verification,
the most important problem is to obtain a witness \emph{finite-memory} winning 
strategy (if one exists).
The winning strategy in the game corresponds to the desired controller for 
(or implementation of) the system in synthesis, and for implementability a finite-memory 
strategy is essential. 
In this work we consider the problem of finite-memory strategy synthesis in multi-dimensional quantitative games in conjunction with parity 
objectives, and the problem of 
existence of memory-efficient randomized strategies for such games. 
These are some of the core and foundational problems in the emerging theory of 
quantitative verification and synthesis.

\smallskip\noindent\textbf{Our contributions.} 
In this work, we give an extended presentation of the results of \cite{chatterjee_CONCUR2012}, the first study of multi-dimensional energy and 
mean-payoff objectives in conjunction with parity objectives. Conjunction of parity objectives with multi-dimensional quantitative objectives had never been considered before \cite{chatterjee_CONCUR2012}. Our presentation is based on the journal publication~\cite{DBLP:journals/acta/ChatterjeeRR14}. Since we consider the synthesis of finite-memory strategies, it follows from the results 
of~\cite{chatterjee_FSTTCS10} that both the problems (multi-dimensional energy with parity and 
multi-dimensional mean-payoff  with parity) are equivalent.
Our main results for finite-memory strategy synthesis for multi-dimensional energy parity 
games are as follows.
{\bf ($i$) Optimal memory bounds}. 
We first show that memory of exponential size is sufficient in multi-dimensional 
energy parity games.
Our result is a significant improvement over the result that can be obtained 
naively from the results known in literature that yields a triple exponential 
bound, even in the case of multi-dimensional energy games without parity. 
Second, we show a matching lower bound by presenting a family of game graphs 
where exponential memory is necessary in multi-dimensional energy games (without parity), 
even when all the transition weights belong to $\Set{-1,0,+1}$. 
Thus we establish {\em optimal memory bounds} for the finite-memory strategy synthesis problem.
{\bf ($ii$) Symbolic and incremental algorithm.}
We present a {\em symbolic} algorithm (in the sense of \cite{DR10}, i.e., using a compact antichain representation of sets by their minimal elements) to compute a finite-memory winning strategy, if one exists, for multi-dimensional energy parity games. Our algorithm is parameterized by the range of energy levels to consider during its execution. So, we can use it in an {\em incremental approach}: first, we search for finite-memory winning strategies with a small range, and increment the range only when necessary. We also establish a bound on the maximal range to consider which ensures completeness of the incremental approach.
In the worst case the algorithm requires exponential time.
Since exponential size memory is required  (and also the decision problem is
coNP-complete~\cite{chatterjee_FSTTCS10}), the worst case exponential bound can be considered as {\em optimal}.
Moreover, as our algorithm is symbolic and incremental, in most relevant problems
in practice, it is expected to be efficient.
{\bf($iii$) Randomized strategies.} 
We also consider when the (pure) finite-memory strategies can be traded off for conceptually much simpler randomized strategies.
We show that for energy objectives randomization is not helpful (as energy objectives
are similar in spirit with safety objectives), even with only one player, neither it is for two-player multi-dimensional mean-payoff objectives. However, randomized memoryless strategies suffice for one-player multi-dimensional mean-payoff parity games. For the important special case of mean-payoff parity objectives (conjunction of a single 
mean-payoff and parity objectives), we show that in games, finite-memory strategies can be traded
off for randomized memoryless strategies.

\smallskip\noindent\textbf{Related works.}
This paper extends the results presented in its preceding conference version~\cite{chatterjee_CONCUR2012} and gives a full presentation of the technical details published in~\cite{DBLP:journals/acta/ChatterjeeRR14}. Games with a single mean-payoff objective have been studied in~\cite{EM79,ZP95},
and games with a single energy objective in~\cite{CdAHS03}; their equivalence was 
established in~\cite{bouyer_FORMATS2008}.
One-dimensional mean-payoff parity games problem has been studied in~\cite{chatterjee_LICS05}: 
an exponential algorithm was given to decide if there exists a winning 
strategy (which in general was shown to require infinite memory); and 
an improved algorithm was presented in~\cite{bouyer_ATVA11}.
One-dimensional energy parity games problem has been studied in~\cite{chatterjee_ICALP10}:
it was shown that deciding the existence of a winning strategy 
is in NP $\cap$ coNP, and an exponential algorithm was given.
It was also shown in~\cite{chatterjee_ICALP10} that, for one-dimensional energy parity objectives,
finite-memory strategies with exponential memory are sufficient, and the
decision problem for mean-payoff parity objective can be reduced to 
energy parity objective. Alternative objectives based on the mean-payoff but with improved tractability in the one-dimensional setting were considered in~\cite{chatterjee_ATVA2013}. Extension of the worst-case threshold problem - the classical decision problem on mean-payoff games - with guarantees on the expected performance faced to a stochastic adversary was studied in~\cite{DBLP:conf/stacs/BruyereFRR14}. 

Games on VASS with several different winning objectives have been studied in~\cite{brazdil_ICALP10}, and
from the results of~\cite{brazdil_ICALP10} it follows that in multi-dimensional energy 
games, winning strategies with finite memory are sufficient (and a triple 
exponential bound on memory can be derived from the results).
The complexity of multi-dimensional energy and mean-payoff games was studied 
in~\cite{chatterjee_FSTTCS10,VR11}.
It was shown in~\cite{chatterjee_FSTTCS10} that in general, winning strategies in multi-dimensional mean-payoff games 
require infinite memory, whereas for multi-dimensional energy games, finite-memory
strategies are sufficient. 
Moreover, for finite-memory strategies, the multi-dimensional mean-payoff and energy 
games coincide, and optimal computational complexity for deciding the existence of a 
winning strategy was established as coNP-complete~\cite{chatterjee_FSTTCS10,VR11}.
Multi-dimensional mean-payoff games with infinite-memory strategies were studied in~\cite{VR11},
and optimal computational complexity results were established.
Various decision problems over multi-dimensional energy games were studied in \cite{fahrenberg_ICTAC11}.

\section{Preliminaries}
\label{preliminaries}

We consider two-player game structures and denote the two \textit{players} by $\playerOne$ and $\playerTwo$.

\smallskip\noindent\textbf{Multi-weighted two-player game structures.} A \textit{multi-weighted two-player game structure} is a tuple \gameFull where (i) \statesOne and \statesTwo resp. denote the finite sets of \textit{states} belonging to \playerOne and $\playerTwo$, with $\statesOne \cap \statesTwo = \emptyset$; (ii) $\initState \in \states = \statesOne \cup \statesTwo$ is the initial state; (iii) $\edges \subseteq \states \times \states$ is the set of \textit{edges} such that for all $s \in \states$, there exists $s' \in \states$ such that $(s, s') \in \edges$; (iv) $k \in \nat$ is the \textit{dimension} of the weight vectors; and (v) $\weight \colon \edges \rightarrow \integ^{\dimension}$ is the multi-weight labeling function. The game structure $\game$ is \textit{one-player} if $\statesTwo = \emptyset$. A \textit{play} in \game is an infinite sequence of states $\play = s_{0}s_{1}s_{2}\ldots{}$ such that $s_{0} = \initState$ and for all $i \geq 0$, we have $(s_{i}, s_{i+1}) \in \edges$. The \textit{prefix} up to the $n$-th state of play $\play = s_{0}s_{1}\ldots{}s_{n}\ldots{}$ is the finite sequence $\play(n) = s_{0}s_{1}\ldots{}s_{n}$. Let $\first(\play(n))$ and $\last(\play(n))$ resp. denote $s_{0}$ and $s_{n}$, the first and last states of $\play(n)$. A prefix $\play(n)$ belongs to $\player_{i}$, $i \in \lbrace 1, 2\rbrace$, if $\last(\play(n)) \in \states_{i}$. The set of plays of \game is denoted by \plays and the corresponding set of prefixes is denoted by $\prefixes$. The set of prefixes that belong to $\player_{i}$ is denoted by $\prefixesArg{i}$. The \textit{energy level vector} of a sequence of states $\prefix = s_{0}s_{1}\ldots{}s_{n}$ such that for all $i \geq 0$, we have $(s_{i}, s_{i+1}) \in \edges$, is \elFull and the \textit{mean-payoff vector} of a play $\play = s_{0}s_{1}\ldots{}$ is \mpayFull.

\smallskip\noindent\textbf{Parity.} A game structure $\game$ is extended with a priority function $\parityFull$ to the structure $\gameParFull$. Given a play $\play = s_{0}s_{1}s_{2}\ldots{}$, we define $\infPlay(\play) = \left\lbrace s \in S \;\vert\; \forall\, m \geq 0, \exists\, n > m \text{ such that } s_{n} = s\right\rbrace$, the set of states that appear infinitely often along $\play$. The \textit{parity} of a play \play is defined as $\minparFull$. In the following definitions, we denote any game by $\gamePar$ with no loss of generality.

\smallskip\noindent\textbf{Strategies.} Given a finite set $A$, a \textit{probability distribution} on $A$ is a function $p \colon A \rightarrow [0, 1]$ such that $\sum_{a\in A} p(a) = 1$. We denote the set of probability distributions on $A$ by $\dist (A)$. A \textit{pure strategy} for $\player_{i}$, $i \in \lbrace 1, 2\rbrace$, in \gamePar is a function $\strat_{i} \colon \prefixesArgPar{i} \rightarrow \states$ such that for all $\prefix \in \prefixesArgPar{i}$, we have $(\last(\prefix), \strat_{i}(\prefix)) \in \edges$. A \textit{(behavioral) randomized strategy} is a function $\strat_{i} \colon \prefixesArgPar{i} \rightarrow \dist(\states)$ such that for all $\prefix \in \prefixesArgPar{i}$, we have $\left\lbrace (\last(\prefix), s) \;\vert\; s \in \states, \strat_{i}(\prefix)(s) > 0\right\rbrace  \subseteq \edges$. A pure strategy $\strat_{i}$ for $\player_{i}$ has \textit{finite memory} if it can be encoded by a deterministic Moore machine $(M,m_0,\alpha_u,\alpha_n)$ where $M$ is a finite set of states (the memory of the strategy), $m_0 \in M$ is the initial memory state, $\alpha_u \colon M \times S \to M$ is an update function, and $\alpha_n \colon M \times S_{i} \to S$ is the next-action function. If the game is in $s \in S_{i}$ and $m \in M$ is the current memory value, then the strategy chooses $s' = \alpha_n(m,s)$ as the next state of the game. When the game leaves a state $s \in S$, the memory is updated to $\alpha_u(m,s)$. Formally, $\left\langle M, m_0, \alpha_u, \alpha_n\right\rangle $ defines the strategy $\strat_{i}$ such that $\strat_{i}(\rho\cdot s) = \alpha_n(\hat{\alpha}_u(m_0, \rho), s)$ for all $\rho \in S^*$ and $s \in S_{i}$, where $\hat{\alpha}_u$ extends $\alpha_u$ to sequences of states as expected. A pure strategy is \emph{memoryless} if $\vert M\vert = 1$, i.e., it does not depend on history but only on the current state of the game. Similar definitions hold for finite-memory randomized strategies, such that the next-action function $\alpha_n$ is randomized, while the update function $\alpha_u$ remains deterministic. We resp. denote by $\strats_{i}, \stratsPureFinite_{i}, \stratsPureMemoryless_{i}, \stratsRandomizedMemoryless_{i}$ the sets of general (i.e., possibly randomized and infinite-memory), pure finite-memory, pure memoryless and randomized memoryless strategies for player $\player_{i}$.

Given a prefix $\prefix \in \prefixesArgPar{i}$ belonging to player $\player_{i}$, and a strategy $\strat_{i} \in \strats_{i}$ of this player, we define the \textit{support} of the probability distribution defined by $\strat_{i}$ as $\supp_{\strat_{i}}(\prefix) = \left\lbrace s \in \states \;\vert\; \strat_{i}(\prefix)(s) > 0\right\rbrace$, with $\strat_{i}(\prefix)(s) = 1$ if $\strat_{i}$ is pure and $\strat_{i}(\prefix) = s$. A play \play is said to be \textit{consistent} with a strategy $\strat_{i}$ of $\player_{i}$ if for all $n \geq 0$ such that $\last(\play(n)) \in \states_{i}$, we have $\last(\play(n+1)) \in \supp_{\strat_{i}}(\play(n))$. Given two strategies, $\strat_{1}$ for $\playerOne$ and $\strat_{2}$ for $\playerTwo$, we define $\outcomePar(\strat_{1}, \strat_{2}) = \left\lbrace \play \in \playsPar \;\vert\; \play \text{ is consistent with } \strat_{1} \text{ and } \strat_{2}\right\rbrace$, the set of possible \textit{outcomes} of the game. Note that if both strategies $\strat_{1}$ and $\strat_{2}$ are pure, we obtain a unique play $\play = s_{0}s_{1}s_{2}\ldots{}$ such that for all $j \geq 0$, $i \in \lbrace 1, 2\rbrace$, if $s_{j} \in \states_{i}$, then we have $s_{j+1} = \strat_{i}(s_{j})$.

Given the initial state $\initState$ and strategies for both players $\strat_{1} \in \strats_{1}$, $\strat_{2} \in \strats_{2}$, we obtain a Markov chain. Thus, every \textit{event} $\event \subseteq \playsPar$, a measurable set of plays, has a uniquely defined probability \cite{vardi_FOCS85} (Carath\'eodory's extension theorem induces a unique probability measure on the Borel $\sigma$-algebra over $\playsPar$). We denote by $\proba_{s_{init}}^{\strat_{1}, \strat_{2}}(\event)$ the probability that a play belongs to $\event$ when the game starts in $s_{init}$ and is played consistently with $\strat_{1}$ and $\strat_{2}$. Let $f : \playsPar \rightarrow \reals$ be a measurable function, we denote $\expect_{s_{init}}^{\strat_{1}, \strat_{2}}(f)$ the expected value of function $f$ over a play when the game starts in $s_{init}$ and is played consistently with $\strat_{1}$ and $\strat_{2}$. We use the same notions for prefixes by naturally extending them to their infinite counterparts.

\smallskip\noindent\textbf{Objectives.} An \textit{objective} for \playerOne in \gamePar is a set of plays $\objective \subseteq \playsPar$. We consider several kinds of objectives:
\begin{itemize}
\item \textit{Multi Energy objectives}. Given an initial natural energy vector $v_{0} \in \nat^{\dimension}$, the objective \objELFull requires that the energy level in all dimensions stays positive at all times.
\item \textit{Multi Mean-payoff objectives}. Given a rational threshold vector $v \in \rat^{\dimension}$, the objective \objMPFull requires that for all dimension $j$, the mean-payoff on this dimension is at least $v(j)$.
\item \textit{Parity objectives}. Objective \objParFull requires that the minimum priority visited infinitely often be even. When the set of priorities is restricted to $\lbrace 0, 1\rbrace$, we have a \textit{Büchi objective}. Note that every multi-weighted game structure $\game$ without parity can trivially be extended to $\gamePar$ with $\parity : \states \rightarrow \left\lbrace 0\right\rbrace$.
\item \textit{Combined objectives}. Parity objectives can naturally be combined with multi mean-payoff and multi energy objectives, resp. yielding $\objMP \cap \objPar$ and $\objEL \cap \objPar$.
\end{itemize}
\smallskip\noindent\textbf{Sure, satisfaction and expectation semantics.} A strategy $\strat_{1}$ for \playerOne is \textit{surely winning} for an objective \objective in \gamePar if for all plays $\play \in \playsPar$ that are consistent with $\strat_{1}$, we have $\play \in \objective$. When at least one of the players plays a randomized strategy, the notion of sure winning in general is too restrictive and inadequate, as the set of consistent plays that do not belong to $\objective$ may have zero probability measure. Therefore, it is useful to use \textit{satisfaction} or \textit{expectation} criteria. Let $\strat_{1} \in\strats_{1}$ be the strategy of $\playerOne$.
\begin{itemize}
\item Given a threshold $\alpha \in \left[ 0, 1\right]$ and a measurable objective $\objective \subseteq \playsPar$, $\alpha$-\textit{satisfaction} asks that for all $\strat_{2} \in \strats_{2}$, we have $\proba_{s_{init}}^{\strat_{1}, \strat_{2}}(\objective) \geq \alpha$. If $\strat_{1}$ satisfies $\objective$ with probability $\alpha = 1$, we say that $\strat_{1}$ is \textit{almost-surely winning} for $\objective$ in $\gamePar$.

\item Given a threshold $\beta \in \rat^{\dimension}$, a function $f : \playsPar \rightarrow \rat$, $\beta$-\textit{expectation} asks that for all $\strat_{2} \in \strats_{2}$, we have $\expect_{s_{init}}^{\strat_{1}, \strat_{2}}(f) \geq \beta$.
\end{itemize}
Note that energy objectives are naturally more enclined towards satisfaction semantics, as they model safety properties.

\smallskip\noindent\textbf{Strategy synthesis problem.} For multi energy parity games, the problem is to synthesize a finite initial credit $v_{0} \in \nat^{\dimension}$ and a pure \textit{finite-memory} strategy $\strat^{pf}_{1} \in \stratsPureFinite_{1}$ that is surely winning for $\playerOne$ in $\gamePar$ for the objective $\objEL \cap \objPar$, \textit{if one exists}. So, the initial credit is not fixed, but is part of the strategy to synthesize.
For multi mean-payoff games, given a threshold $v \in \rat^{\dimension}$, the problem is to synthesize a pure \textit{finite-memory} strategy $\strat^{pf}_{1} \in \stratsPureFinite_{1}$ that is surely winning for $\playerOne$ in $\gamePar$ for the objective $\objMP \cap \objPar$, \textit{if one exists}.
Note that multi energy and multi mean-payoff games are equivalent for finite-memory strategies, while in general, infinite memory may be necessary for the latter \cite{chatterjee_FSTTCS10}. 

\smallskip\noindent\textbf{Trading finite memory for randomness.}
We study when finite memory can be traded for randomization. The question is: given a strategy $\strat_{1}^{pf} \in \stratsPureFinite_{1}$ which ensures surely winning of some objective $\objective$, does there exist a strategy $\strat_{1}^{rm} \in \stratsRandomizedMemoryless_{1}$ which ensures almost-surely winning for the same objective $\objective$? For mean-payoff objectives, one can also ask for a weaker equivalence, that is: can randomized memoryless strategies achieve the same expectation as pure finite-memory ones?

\section{Optimal memory bounds}
\label{parityChanges}

In this section, we establish optimal memory bounds for pure finite-memory winning strategies on multi-dimensional energy parity games (MEPGs). Also, as a corollary, we obtain results for pure finite-memory winning strategies on multi-dimensional mean-payoff parity games (MMPPGs). We show that single exponential memory is both sufficient and necessary for winning strategies. Additionally, we show how the parity condition in a MEPG can be removed by adding additional energy dimensions.

\paragraph{{\bf Multi energy parity games.}} A sample game is depicted on Fig. \ref{fig:gameAndTree}. The key point in the upper bound proof on memory is to understand that for $\playerOne$ to win a multi energy parity game, he must be able to force cycles whose energy level is positive in all dimensions and whose minimal parity is even. As stated in the next lemma, finite-memory strategies are sufficient for multi energy parity games for both players.

\vspace{-2mm}
\begin{figure}[htb]
  \centering   
  \begin{minipage}[r]{.35\linewidth}
   \hspace*{6mm}\scalebox{1}{\begin{tikzpicture}[->,>=stealth',shorten >=1pt,auto,node
    distance=2.5cm,bend angle=45, scale=0.5, font=\scriptsize]
    \tikzstyle{p1}=[draw,circle,text centered,minimum size=7mm, text width = 5mm]
    \tikzstyle{p2}=[draw,rectangle,text centered,minimum size=7mm, text width = 5mm]
    \node[p2]  (0)  at (0, 0) {$s_{0}$\\$2$};
    \node[p2]  (1) at (-2, -3) {$s_{1}$\\$3$};
    \node[p2]  (2) at (2, -3)  {$s_{2}$\\$1$};
    \node[p1]  (3) at (0, -6)  {$s_{3}$\\$2$};
    \node[p1]  (4)  at (-2, -9) {$s_{4}$\\$3$};
    \node[p1]  (5)  at (2, -9) {$s_{5}$\\$0$};
    \coordinate[shift={(0mm,5mm)}] (init) at (0.north);
    \path
    (init) edge (0);
	\draw[->,>=latex] (0) to node[left] {$(-1,1)$} (1);
	\draw[->,>=latex] (0) to node[right] {$(0,2)$} (2);
	\draw[->,>=latex] (1) to node[left] {$(0,1)$} (3);
	\draw[->,>=latex] (2) to node[right] {$(0,0)$} (3);
	\draw[->,>=latex] (3) to node[left] {$(1,-1)$} (4);
	\draw[->,>=latex] (3) to node[right, xshift=-0.5mm] {$(-2,1)$} (5);
	\draw[->,>=latex] (4) to[out=135,in=180] node[above, xshift=-5mm] {$(0,-1)$} (0);
	\draw[->,>=latex] (5) to[out=30,in=0] node[above] {$\quad\quad\;(2,0)$} (3);
      \end{tikzpicture}}
   \end{minipage}
   \hspace*{10mm}\begin{minipage}[l]{.47\linewidth}
  \scalebox{1}{\begin{tikzpicture}[->,>=stealth',shorten >=1pt,auto,node
    distance=2.5cm,bend angle=45, scale=0.4, font=\scriptsize]
    \tikzstyle{p1}=[draw,ellipse,text centered,minimum width=24mm]
    \tikzstyle{p2}=[draw,rectangle,text centered,minimum width=18mm]
    \node[p2]  (0)  at (0, 0) {$\langle s_{0}, (0, 0)\rangle$};
    \node[p2]  (1) at (-3.5, -3) {$\langle s_{1}, (-1, 1)\rangle$};
    \node[p2]  (2) at (3.5, -3)  {$\langle s_{2}, (0, 2)\rangle$};
    \node[p1]  (3) at (-3.5, -6)  {$\langle s_{3}, (-1, 2)\rangle$};
    \node[p1]  (4)  at (3.5, -6) {$\langle s_{3}, (0, 2)\rangle$};
    \node[p1]  (5)  at (-3.5, -9) {$\langle s_{4}, (0, 1)\rangle$};
    \node[p1]  (6)  at (3.5, -9) {$\langle s_{5}, (-2, 3)\rangle$};
    \node[p2]  (7)  at (-3.5, -12) {$\langle s_{0}, (0, 0)\rangle$};
    \node[p1]  (8)  at (3.5, -12) {$\langle s_{3}, (0, 3)\rangle$};
    \coordinate[shift={(0mm,5mm)}] (init) at (0.north);
    \path
    (init) edge (0)
    (0) edge (1)
    (0) edge (2);
	\draw[->,>=latex] (1) to[out=270,in=90] (3);
	\draw[->,>=latex] (2) to[out=270,in=90] (4);
	\draw[->,>=latex] (3) to[out=270,in=90] (5);
	\draw[->,>=latex] (4) to[out=270,in=90] (6);
	\draw[->,>=latex] (5) to[out=270,in=90] (7);
	\draw[->,>=latex] (6) to[out=270,in=90] (8);
	\draw[->,dashed,>=latex] (7) to[out=180,in=180] (0);
	\draw[->,dashed,>=latex] (8) to[out=5,in=355] (4);
      \end{tikzpicture}}
   \end{minipage}
      \caption{Two-dimensional energy parity game and even-parity self-covering tree representing an arbitrary finite-memory winning strategy. Circle states belong to $\playerOne$, square states to $\playerTwo$.}
\label{fig:gameAndTree}
  \end{figure}
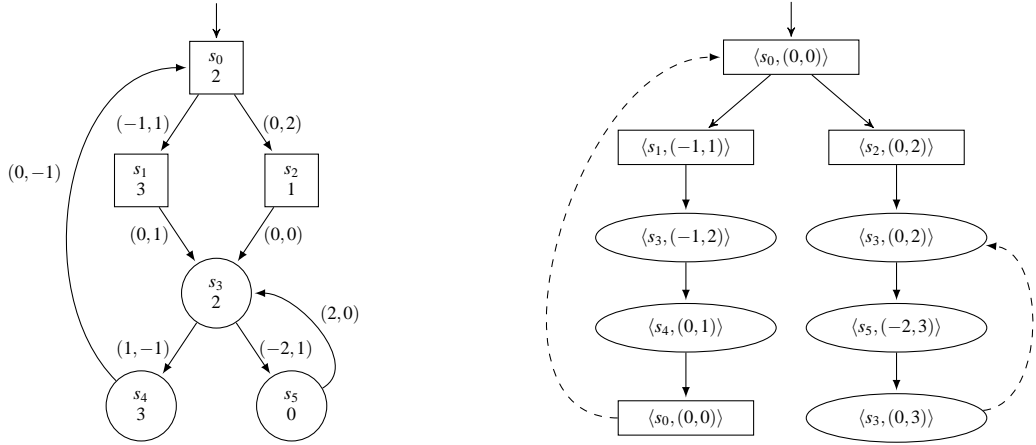

\begin{lemma}[{Extension of \cite[Lemma 2 and 3]{chatterjee_FSTTCS10}}]
\label{purefinitememory}
If $\playerOne$ has a winning strategy in a multi energy parity game, then he has a pure finite-memory one. If $\playerTwo$ has a winning strategy in a multi energy parity game, then he has a pure memoryless one.
\end{lemma}

\begin{proof}
The first part of the result follows using the standard well-quasi ordering argument (straightforward extension of \cite[Lemma 2]{chatterjee_FSTTCS10}). The second part follows by the classical edge induction argument: Lemma 3 of \cite{chatterjee_FSTTCS10} and Lemma 3 of \cite{chatterjee_ICALP10} show the result using edge induction for multi energy and energy parity games, respectively. Repeating the arguments of Lemma 3 of \cite{chatterjee_ICALP10}, and replacing the part on single
energy objectives by the argument of Lemma 3 of \cite{chatterjee_FSTTCS10} for multi energy objectives, we obtain the desired result.\qed
\end{proof}

By Lemma~\ref{purefinitememory}, we know that w.l.o.g.~both players can be restricted to play pure finite-memory strategies. The property on the cycles can then be formalized as follows.
\begin{lemma}
\label{lemma_pumpCycle}
Let $\gameParFull$ be a multi energy parity game. Let $\strat_{1}^{pf} \in \stratsPureFinite_{1}$ be a winning strategy of $\playerOne$ for initial credit $v_{0} \in \nat^{k}$. Then, for all $\strat_{2}^{pm} \in \stratsPureMemoryless_{2}$, the outcome is a \textit{regular} play $\play = \prefix \cdot (\finplay_{\infty})^{\omega}$, with $\prefix \in \prefixes, \finplay_{\infty} \in \states^{+}$, such that $\el(\finplay_{\infty}) \geq 0$ and $\minpar(\play) = \min \left\lbrace \parity(s) \;\vert\; s \in \finplay_{\infty} \right\rbrace$ is even.
\end{lemma}

\begin{proof}
Recall that both players play with pure finite memory strategies. Therefore, a finite number of decisions are made and the outcome is a regular play $\play = \prefix \cdot (\finplay_{\infty})^{\omega}$. Note that $\el(\prefix)$ does not have to be positive, as $\playerOne$ may have $v_{0} > \el(\prefix)$. Similarly, priorities of states visited in $\prefix$ have no impact on winning as they are only visited a finite number of times. First, suppose $\el(\finplay_{\infty}) < 0$ on some dimension $1 \leq j \leq \dimension$. Then, after $m > 0$ cycles, for some $n > 0$, the energy level will be $\el(\play(n)) = \el(\prefix \cdot (\finplay_{\infty})^{m}) = \el(\prefix) + m \cdot \el(\finplay_{\infty})$. Since $v_{0}$ is finite and $m \rightarrow \infty$, there exist some $m, n > 0$, such that $v_{0} + \el(\play(n)) < 0$ on dimension $j$ and $\strat_{1}$ is not winning. Second, suppose $\min \left\lbrace \parity(s) \;\vert\; s \in \finplay_{\infty} \right\rbrace$ is odd. Since the set of states visited infinitely often is exactly the set of states in $\finplay_{\infty}$, this implies that $\minpar(\play)$ is odd, and thus $\strat_{1}$ is not winning.\qed
\end{proof}

A {\em self-covering path} in a game, straightforwardly extending the notion introduced by Rackoff \cite{rackoff_TCS78} for \textit{Vector Addition Systems} (VAS), is a sequence of states $\state_{0}\state_{1}\state_{2}\ldots{}\state_{m}$ such that there exist two positions $i$ and $j$ that verify $0 \leq i < j \leq m$, $s_{i} = s_{j}$ and $\el(\state_{0}\ldots{}\state_{i}) \leq \el(\state_{0}\ldots{}\state_{i}\ldots{}\state_{j})$. In other words, such a path describes a finite prefix followed by a cycle which has a non-negative effect on the energy level. Ensuring such cycles is crucial to win the energy objective. With the notion of regular play of Lemma \ref{lemma_pumpCycle}, we generalize the notion of self-covering path to include the parity condition.
We show here that, if such a path exists, then the lengths of its cycle and the prefix needed to reach it can be bounded. Bounds on the strategy follow. In \cite{rackoff_TCS78}, Rackoff showed how to bound the length of self-covering paths in VAS. This work was extended to Vector Addition Systems with States (VASS) by Rosier and Yen \cite{rosier_JCSS86}. Recently, Br{\'a}zdil \textit{et al.} introduced reachability games on VASS and the notion of \textit{self-covering trees} \cite{brazdil_ICALP10}. Their Zero-safety problem with $\omega$ initial marking is equivalent to multi energy games with weights in $\lbrace -1, 0, 1\rbrace$, and without the parity condition.  They showed that if winning strategies exist for $\playerOne$, then some of them can be represented as \textit{self-covering trees} of bounded depth. Trees have to be considered instead of paths, as in a game setting all the possible choices of the adversary ($\playerTwo$) must be considered. Here, we extend the notion of self-covering trees to \textit{even-parity self-covering trees}, in order to handle parity objectives.

\paragraph{{\bf Even-parity self-covering tree.}}
An \textit{even-parity self-covering tree} (\epSCT) for $s \in \states$ is a finite tree $\treeFull$, where $\treeStates$ is the set of nodes, $\lab \colon \treeStates \rightarrow S \times \integ^{\dimension}$ is a labeling function and $\treeEdges \subset \treeStates \times \treeStates$ is the set of edges, such that
\begin{itemize}
\item[$\bullet$] The root of $\tree$ is labeled $\langle s, (0, \ldots{}, 0)\rangle$.
\item[$\bullet$] If $\varsigma \in \treeStates$ is not a leaf, then let $\lab(\varsigma) = \langle t, u \rangle$, $t \in \states$, $u \in \integ^{\dimension}$, such that
\begin{itemize}
\item[-] if $t \in \states_{1}$, then $\varsigma$ has a unique child $\otherNode$ such that $\lab(\otherNode) = \langle t', u'\rangle$, $(t, t') \in E$ and $u' = u + \weight(t, t')$;
\item[-] if $t \in \states_{2}$, then there is a bijection between children of $\varsigma$ and edges of the game leaving $t$, such that for each successor $t' \in \states$ of $t$ in the game, there is one child $\otherNode$ of $\varsigma$ such that $\lab(\otherNode) = \langle t', u'\rangle$, $u' = u + \weight(t, t')$.
\end{itemize}
\item[$\bullet$] If $\varsigma$ is a leaf, then let $\lab(\varsigma) = \langle t, u \rangle$ such that there is some ancestor $\otherNode$ of $\varsigma$ in $\tree$ such that $\lab(\otherNode) = \langle t, u' \rangle$, with $u' \leq u$, and the downward path from $\otherNode$ to $\varsigma$, denoted by $\otherNode \ancPath \node$, has minimal priority even. We say that $\otherNode$ is an \textit{even-descendance energy ancestor} of $\varsigma$.
\end{itemize}

Intuitively, each path from root to leaf is a self-covering path of even parity in the game graph so that plays unfolding according to such a tree correspond to winning plays of Lemma \ref{lemma_pumpCycle}. Thus, the epSCT fixes how $\playerOne$ should react to actions of $\playerTwo$ in order to win the MEPG (Fig. \ref{fig:gameAndTree}). Note that as the tree is finite, one can take the largest negative number that appears on a node in each dimension to compute an initial credit for which there is a winning strategy (i.e., the one described by the tree). In particular, let $\largestW$ denote the maximal absolute weight appearing on an edge in $\gamePar$. Then, for an epSCT $\tree$ of depth $\depth$, it is straightforward to see that the maximal initial credit required is at most $\depth \cdot \largestW$ as the maximal decrease at each level of the tree is bounded by $\largestW$. We suppose $\largestW > 0$ as otherwise, any strategy of $\playerOne$ is winning for the energy objective, for any initial credit vector $v_{0} \in \nat^{k}$.

Let us explicitely state how $\playerOne$ can deploy a strategy $\strat_{1}^{\tree} \in \stratsPureFinite_{1}$ based on an epSCT $\treeFull$. We refer to such a strategy as an \textit{epSCT strategy}. It consists in following a path in the tree $\tree$, moving a pebble from node to node and playing in the game depending on edges taken by this pebble. Each time a node $\varsigma$ such that $\lab(\node) = \langle t, u \rangle$ is encountered, we do the following.
\begin{itemize}
\item[$\bullet$] If $\varsigma$ is a leaf, the pebble directly goes up to its oldest even-descendance energy ancestor $\otherNode$. By oldest we mean the first encountered when going down in the tree from the root. Note that this choice is arbitrary, in an effort to ease following proof formulations, as any one would suit.
\vspace{5mm}
\item[$\bullet$] Otherwise, if $\varsigma$ is not a leaf,
\begin{itemize}
\item[-] if $t \in \states_{2}$ and $\playerTwo$ plays state $t' \in S$, the pebble is moved along the edge going to the only child $\otherNode$ of $\varsigma$ such that $\lab(\otherNode) = \langle t', u' \rangle$, $u' = u + \weight(t, t')$;
\item[-] if $t \in \states_{1}$, the pebble moves to $\otherNode$, $\lab(\otherNode) = \langle t', u' \rangle$, the only child of $\varsigma$, and $\playerOne$ strategy is to choose the state $t'$ in the game.
\end{itemize}
\end{itemize}
If such an epSCT $\tree$ of depth $\depth$ exists for a game $\gamePar$, then $\playerOne$ can play the strategy $\strat_{1}^{\tree} \in \stratsPureFinite_{1}$ to win the game with initial credit bounded by $\depth \cdot W$.

\paragraph{{\bf Bounding the depth of epSCTs.}}
Consider a multi energy game \textit{without} parity. Then, the priority condition on downward paths from ancestor to leaf is not needed and self-covering trees (i.e., epSCTs without the condition on priorities) suffice to describe winning strategies. One can bound the size of SCTs using results on the size of solutions for linear diophantine equations (i.e., with integer variables) \cite{borosh_AMS76}. In particular, recent work on reachability games over VASS with weights $\left\lbrace -1, 0, 1\right\rbrace$, Lemma 7 of \cite{brazdil_ICALP10}, states that if $\playerOne$ has a winning strategy on a VASS, then he can exhibit one that can be described as an SCT whose \textit{depth} is at most $l = 2^{(d-1) \cdot \vert\states\vert} \cdot (\vert\states\vert + 1)^{c \cdot \dimension^{2}}$, where $c$ is a constant independent of the considered VASS and $d$ its branching degree (i.e., the highest number of outgoing edges on any state). Naive use of this bound for multi energy games with arbitrary integer weights would induce a \textit{triple} exponential bound for memory. Indeed, recall that $\largestW$ denotes the maximal absolute weight that appears in a game $\gameParFull$. A straightforward translation of a game with arbitrary weights into an equivalent game that uses only weights in $\{-1,0,1\}$ induces a blow-up by $\largestW$ in the size of the state space, and thus an exponential blow-up by $\largestW$ in the depth of the tree, which becomes doubly exponential as we have 
\begin{equation*}
l = 2^{(\branching-1) \cdot \largestW \cdot \vert\states\vert} \cdot (\largestW \cdot \vert\states\vert + 1)^{c \cdot \dimension^{2}} = 2^{(\branching-1) \cdot 2^{\largestWenc}\cdot \vert\states\vert} \cdot (\largestW \cdot \vert\states\vert + 1)^{c \cdot \dimension^{2}},
\end{equation*}
where $\largestWenc$ denotes the number of bits used by the encoding of $\largestW$. Moreover, the width of the tree increases as $\branching^{l}$, i.e., it increases exponentially with the depth. So straight application of previous results provides an overall tree of triple exponential size. In this paper we improve this bound and prove a single exponential upper bound, even for multi energy \textit{parity} games. We proceed in two steps, first studying the depth of the epSCT, and then showing how to compress the tree into a \textit{directed acyclic graph} (DAG) of \textit{single} exponential size.

\begin{lemma}
\label{lemma_depth_parity}
Let $\gameParFull$ be a multi energy parity game such that $\largestW$ is the maximal absolute weight appearing on an edge and $\branching$ the branching degree of $\gamePar$. Suppose there exists a finite-memory winning strategy for $\playerOne$. Then there is an even-parity self-covering tree for $\initState$ of depth at most $\depth = 2^{(\branching-1) \cdot \vert\states\vert} \cdot \left( \largestW \cdot \vert\states\vert + 1\right) ^{c \cdot \dimension^{2}}$, where $c$ is a constant independent of $\gamePar$.
\end{lemma}

Lemma \ref{lemma_depth_parity} eliminates the exponential blow-up in depth induced by a naive coding of arbitrary weights into $\{-1,0,1\}$ weights, and implies an overall doubly exponential upper bound. Our proof is a generalization of \cite[Lemma 7]{brazdil_ICALP10}, using a more refined analysis to handle both \textit{parity} and \textit{arbitrary integer weights}. The idea is the following. First, consider the one-player case. The epSCT is reduced to a path. By Lemma \ref{lemma_pumpCycle}, it is composed of a finite prefix, followed by an infinitely repeated sequence of positive energy level and even minimal priority. The point is to bound the length of such a sequence by eliminating cycles that are not needed for energy or parity. Second, to extend the result to two-player games, we use an induction on the number of choices available for $\playerTwo$ in a given state. Intuitively, we show that if $\playerOne$ can win with an epSCT $\tree_{A}$ when $\playerTwo$ plays edges from a set $A$ in a state $s$, and if he can also win with an epSCT $\tree_{B}$ when $\playerTwo$ plays edges from a set $B$, then he can win when $\playerTwo$ chooses edges from both $A$ and $B$, with an epSCT whose depth is bounded by the sum of depths of $\tree_{A}$ and $\tree_{B}$.

\begin{proof}
The proof is made in two steps. First, we consider the one-player case, where $\states_{2} = \emptyset$. Second, we use an induction scheme over the choice degree of $\playerTwo$ to extend our results to the two-player case.

We start with $\states_{2} = \emptyset$, the one-player game. By Lemma \ref{lemma_pumpCycle}, a winning play is of the form $\play = \prefix \cdot (\finplay_{\infty})^{\omega}$ such that $\el(\finplay_{\infty}) \geq 0$ and $\minpar(\play) = \min \left\lbrace \parity(s) \;\vert\; s \in \finplay_{\infty} \right\rbrace$ is even. Notice that such a play corresponds to the epSCT defined above, as it reduces to an even-parity self-covering path $\langle s_{init}, (0, \ldots, 0) \rangle \ancPath \langle s, u\rangle \ancPath \langle s, u'\rangle$ with $u' \geq u$. Therefore its existence is guaranteed and it remains to bound its length. Given such a path, the idea is to eliminate unnecessary cycles, in order to reduce its length while maintaining the needed properties (i.e., positive energy and even minimal priority). First, notice that cycles in the sub-path $\langle s_{init}, (0, \ldots, 0) \rangle \ancPath \langle s, u\rangle$ can be trivially erased as they are only visited a finite number of times and thus (a) the initial credit can compensate for the loss of their potential positive energy effect, and (b) they do not contribute in the parity. Now consider the sub-path $\langle s, u\rangle \ancPath \langle s, u'\rangle$. Since it induces a winning play, its minimal priority is even. Let $p_{m}$ be this priority. We may suppose w.l.o.g. that $\parity(s) = p_{m}$, otherwise it suffices to shift this sub-path to $\langle s', v\rangle \ancPath \langle s', v'\rangle$ for some state $s'$ such that $\parity(s') = p_{m}$ and $v' \geq v$, and add the sub-path $\langle s, u\rangle \ancPath \langle s', v\rangle$ to the finite prefix. Now we may eliminate each cycle of $\langle s, u\rangle \ancPath \langle s, u'\rangle$ safely in regards to the parity objective as they only contain states with greater or equal priority. Thus, we only need to take care of the energy, and fall under the scope of \cite[Lemma 15]{brazdil_ICALP10} for the special case of weights in $\left\lbrace -1, 0, 1\right\rbrace$, where an upper bound $h\left(\vert \states\vert, \dimension\right) = \left(\vert \states\vert + 1\right)^{c \cdot \dimension^{2}}$ on the length of such a path is shown.

We claim that for a one-player game $\game$, with weights in $\left\lbrace -\largestW, -\largestW +1, \ldots{}, \largestW -1, \largestW\right\rbrace$, an upper bound $h\left(\largestW, \vert \states\vert, \dimension\right) = \left(\largestW \cdot \vert\states\vert + 1 \right)^{c \cdot \dimension^{2}}$ is obtained. Indeed, one can translate $\gameParFull$ into an equivalent game $\game'_{\parity'} = \left( \statesOne', \statesTwo, \initState, \edges', \dimension, \weight', \parity'\right)$ such that each edge of $\gamePar$ is split into at most $\largestW$ edges in $\game'_{\parity'}$, with at most $(\largestW -1)$ dummy states in between, so that each edge of $\game'_{\parity'}$ only uses weights in $\left\lbrace -1, 0, 1\right\rbrace$. Let $\states_{d}$ denote the set of these added dummy states. We define this translation $\translation \colon \gamePar \mapsto \game'_{\parity'}$ with $\translation(\statesOne) = \statesOne \cup \states_{d}$, $\translation(\statesTwo) = \statesTwo$, $\translation(\initState) = \initState$, $\translation(\edges) = \bigcup_{(s, t) \in \edges} \translation((s,t))$, $\translation(\dimension)  = \dimension$, $\translation(\weight) = \weight' \colon \edges' \rightarrow \left\lbrace -1, 0, 1\right\rbrace^{\dimension}$, $\translation(\parity) = \parity' \colon \states' \rightarrow \nat$ such that for all $(s, t) \in \edges$ such that $m = \max \left\lbrace \weight(s, t) (j) \;\vert\; 1 \leq j \leq \dimension\right\rbrace - 1$, we have that $\translation\left( (s, t)\right) = \left\lbrace (s, s_{d}^{1}), (s_{d}^{1}, s_{d}^{2}), \ldots{}, (s_{d}^{m-1}, s_{d}^{m}), (s_{d}^{m}, t)\right\rbrace$ such that 
\begin{equation*}
\Big(\forall\; j > 0,\; s_{d}^{j} \in \states_{d} \;\wedge\; p'(s_{d}^{j}) = p(s)\Big) \;\wedge\; \sum_{(q, r) \in \translation\left( (s, t)\right)} \weight' (q, r) = \weight (s, t).
\end{equation*}
To be formally correct, we have to add that for all $s_{d} \in S_{d}$, we have $\degree_{in}(s_{d}) = \degree_{out}(s_{d}) = 1$, and for all $s \not\in \states_{d}$, we have $p'(s) = p(s)$. This translation does not hinder the outcome of the game as each edge in $\gamePar$ has a unique corresponding path in $\game'_{\parity'}$ that preserves the weights and the visited priorities, and that offers no added choice to $\playerOne$. Since $\gamePar$ possesses $\vert\edges\vert \leq \vert\states\vert^{2}$ edges, and for each edge of $\gamePar$, we add at most $(\largestW -1)$ dummy states in $\game'_{\parity'}$, we have $\vert\states'\vert \leq \vert\states\vert + \vert\states\vert^{2}\cdot (\largestW -1) \leq \vert\states\vert^{2}\cdot \largestW$. Therefore, by applying \cite[Lemma 15]{brazdil_ICALP10} on $\game'_{\parity'}$, we obtain the following upper bound:
\begin{equation*}
h\left(W, \vert \states\vert, \dimension\right) = h\left(\vert \states'\vert, \dimension\right) \leq \left(\vert \states\vert^{2}\cdot \largestW + 1\right)^{c \cdot \dimension^{2}} \leq \left(\largestW \cdot \vert \states\vert + 1\right)^{c' \cdot \dimension^{2}}
\end{equation*}
for some constant $c'$ that is independent of $\gamePar$.

Now, consider $\states_{2} \neq \emptyset$. (I) We extend \cite[Lemma 16]{brazdil_ICALP10} for parity. This will help us to establish an induction scheme over the choice degree of $\playerTwo$. Suppose $s \in \statesTwo$ has more than one outgoing edge. Let $\tau = (s,t) \in \edges$ be one of them and $R \subset \edges$ denote the nonempty set of other outgoing edges. Let $\gamePar^{\tau}$ (resp. $\gamePar^{R}$) be the game induced when removing $R$ (resp. $\tau$) from $\gamePar$. Suppose that (a) $s$ is winning for $\playerOne$ in $\gamePar^{R}$ for initial credit $v_{R} \in \nat^{\dimension}$, and (b) there exists some state $s' \in \states$ such that $s'$ is winning for $\playerOne$ in $\gamePar^{\tau}$ for initial credit $v_{\tau} \in \nat^{\dimension}$. We claim that $s'$ is winning in $\gamePar$ for initial credit $v_{0} = v_{\tau} + v_{R}$. Indeed, let $\strat_{1}^{\tau}$ and $\strat_{1}^{R}$ resp. denote winning strategies for $\playerOne$ in $\gamePar^{\tau}$ and $\gamePar^{R}$. Let $\playerOne$ use the following strategy. Player $\playerOne$ plays $\strat_{1}^{\tau}$ as long as $\playerTwo$ does not play any edge of $R$. If such an edge is played, then $\playerOne$ switches to strategy $\strat_{1}^{R}$ and plays it until edge $\tau$ is played again by $\playerTwo$, in which case $\playerOne$ switches back to $\strat_{1}^{\tau}$, and so on. In this way, the outcome of the game is guaranteed to be a play $\play = s' \ldots{} s \ldots{} s \ldots{} s \ldots{} $ resulting from a merge between a play consistent with $\strat_{1}^{\tau}$ over $\gamePar^{\tau}$ (whose energy level is bounded by $-v_{\tau}$ at all times), and a play consistent with $\strat_{1}^{R}$ over $\gamePar^{R}$ (whose energy level is bounded by $-v_{R}$ at all times). Therefore, the combined overall energy level of any prefix $\prefix$ of this play is bounded by $(-v_{\tau} - v_{R})$ as positive cycles in $\gamePar^{\tau}$ and $\gamePar^{R}$ do remain positive in $\gamePar$. Furthermore, the parity condition is preserved in $\gamePar$. Indeed, suppose it is not. Thus, there exists a state visited infinitely often in the outcome such that its priority is minimal and odd. However, as the outcome results from merging plays resp. consistent with $\strat_{1}^{\tau}$ and $\strat_{1}^{R}$, this implies that one of those strategies yields an odd minimal priority, which contradicts the fact that they are winning. This proves the claim.

(II) We apply the induction scheme of \cite[Lemma 18]{brazdil_ICALP10} on $r = \vert\lbrace(s,t) \in \edges \;\vert\; s \in \states_{2}\rbrace\vert - \vert\states_{2}\vert \leq (d-1) \cdot \vert\states\vert$, the choice degree of $\playerTwo$. Notice that our translation $\translation \colon \gamePar \mapsto \game'_{\parity'}$ maintains this choice degree unchanged. The claim is that for a winning state $s'$, there is an epSCT of depth bounded by $2^{r}\cdot h(\largestW, \vert\states\vert, \dimension)$. We have proved that for the base case $r = 0$, similar to $\states_{2} = \emptyset$, this claim is true. So assume it holds for $r$, it remains to prove that it is preserved for $r + 1$. Let $s \in \states_{2}$ be such that $\playerTwo$ has at least two outgoing edges. As before, we define $\gamePar^{\tau}$ and $\gamePar^{R}$. Clearly, the choice degree of $\playerTwo$ is at most $r$ in both games. Let $s'$ be a winning state in $\gamePar$. As $\playerTwo$ has less choices in both $\gamePar^{\tau}$ and $\gamePar^{R}$, clearly $s'$ is still winning in those games. If an epSCT in either of them (which are guaranteed to exist and have depth bounded by $2^{r} \cdot h(\largestW, \vert\states\vert, \dimension)$ by hypothesis) do not contain the state $s$, then the claim is verified. Now suppose we have two epSCTs for games $\gamePar^{\tau}$ and $\gamePar^{R}$ such that they both contain state $s$. Notice that $s$ is winning in those two games and as such, is the root of two respective epSCTs of depth less than $2^{r} \cdot h(\largestW, \vert\states\vert, \dimension)$. Applying (I) on states $s'$ and $s$, we get an epSCT for $s'$ in $\gamePar$ of depth $2 \cdot 2^{r} \cdot h(\largestW, \vert\states\vert, \dimension)$, which concludes the proof.\qed
\end{proof}

\paragraph{{\bf From multi energy parity games to multi energy games.}}
Let $\gamePar$ be a MEPG and assume that $\playerOne$ has a winning strategy in that game. By Lemma \ref{lemma_depth_parity}, there exists an epSCT whose depth is bounded by $\depth$. 
As a direct consequence of that bounded depth, we have that $\playerOne$, by playing the strategy prescribed by the epSCT, enforces a stronger objective than the parity objective. Namely, this strategy ensures to ``never visit more than $\depth$ states of odd priorities before seeing a smaller even priority'' (which is a safety objective). Then, the parity condition can be transformed into additional energy dimensions. 

While our transformation shares ideas with the classical transformation of parity objectives into safety objectives, first proposed in \cite{BernetJW02} (see also \cite[Lemma 6.4]{DR11}), it is technically different because energy levels cannot be reset (as it would be required by those classical constructions). The reduction is as follows. For each odd priority, we add one dimension. The energy level in this dimension is decreased by $1$ each time this odd priority is visited, and it is increased by $\depth$ each time a smaller even priority is visited. If $\playerOne$ is able to maintain the energy level positive for all dimensions (for a given initial energy level), then he is clearly winning the original parity objective; on the other hand, an epSCT strategy that wins the original objective also wins the new game.

\begin{lemma}
\label{lemma_parityToEnergy}
Let $\gameParFull$ be a multi energy parity game with priorities in $\{ 0,1,\dots,2 \cdot \priorities\}$, such that $\largestW$ is the maximal absolute weight appearing on an edge. Then we can construct a multi energy game $\game$ with the same set of states, $(\dimension + \priorities)$ dimensions and a maximal absolute weight bounded by $l$, as defined by Lemma \ref{lemma_depth_parity}, such that $\playerOne$ has a winning strategy in $\game$ iff he has one in $\gamePar$.
\end{lemma}
\begin{proof}
Let $\gameParFull$ be a MEPG with priorities in $\{ 0,1,\dots,2\cdot\priorities\}$. Let $\game = \left( \statesOne, \statesTwo, \edges, (\dimension + \priorities), \weight'\right)$ be the multi energy game\index{multi energy game} (MEG)\index{MEG|see{multi energy game}} obtained by applying the following transformation: $\forall\; (s, t) \in \edges$, $\forall\; 1 \leq j \leq \dimension$, $\weight'((s,t))(j) = \weight((s,t))(j)$, and (a) if $p(t)$ is even, $\forall\; \dimension < j \leq \dimension + \frac{p(t)}{2}$, $\weight'((s,t))(j) = 0$ and $ \forall\; \dimension  + \frac{p(t)}{2} < j \leq \dimension + \priorities$, $\weight'((s,t))(j) = l$, or (b) if $p(t)$ is odd, $\forall\; \dimension < j \leq \dimension + \priorities$, $j \neq \dimension + \left\lceil\frac{p(t)}{2}\right\rceil$, $\weight'((s,t))(j) = 0$ and $\weight'((s,t))(\dimension + \left\lceil\frac{p(t)}{2}\right\rceil) = -1$. We have to prove both ways of the equivalence.

First, suppose $\strat_{1} \in \stratsPureFinite_{1}$ is a winning strategy for $\playerOne$ in the MEPG $\gamePar$. By Lemma \ref{lemma_depth_parity}, there is an epSCT of depth at most $l$ for $s_{init}$. Thus, we know that in every repeated sequence of $l$ states, the minimal visited priority will be even. Therefore, for all additional dimensions, ranging from $\dimension + 1$ to $\dimension + \priorities$, the effect of a sequence of $l$ states will be bounded from below by $-1 \cdot (l-1) + l$, which is positive. Thus strategy $\strat_{1}$ is also winning in $\game$ (with initial credit bounded by $l$ on additional dimensions).

Second, suppose $\strat_{1} \in \stratsPureFinite_{1}$ is a winning strategy for $\playerOne$ in the MEG $\game$, as defined above. Since $\strat_{1}$ is winning, it yields an SCT (epSCT without the parity condition) of bounded depth such that $\playerOne$ is able to enforce positive energy cycles. By definition of weights over $\game$, this cannot be the case if the minimal priority infinitely often visited is odd. Thus this strategy is winning for parity on $\gamePar$, and stays winning for energy over dimensions $1$ to $\dimension$ as weights are unchanged.\qed
\end{proof}

\paragraph{{\bf Bounding the width.}} Thanks to Lemma \ref{lemma_parityToEnergy}, we continue with multi energy games without parity. In order to bound the overall size of memory for winning strategies, we consider the width of self-covering trees. The following lemma states that SCTs, whose width is at most doubly exponential by application of Lemma \ref{lemma_depth_parity}, can be compressed into \textit{directed acyclic graphs} (DAGs) of single exponential width. Thus we eliminate the second exponential blow-up and give an overall single exponential bound for memory of winning strategies. 

\begin{lemma}
\label{lemma_width}
Let $\gameFull$ be a multi energy game such that $\largestW$ is the maximal absolute weight appearing on an edge and $d$ the branching degree of $\game$. Suppose there exists a finite-memory winning strategy for $\playerOne$. Then, there exists $\strat^{\DAG}_{1} \in \stratsPureFinite_{1}$ a winning strategy for $\playerOne$ described by a DAG $\DAG$ of depth at most $\depth = 2^{(d-1) \cdot \vert\states\vert} \cdot \left( \largestW \cdot \vert\states\vert + 1\right) ^{c \cdot \dimension^{2}}$ and width at most $\width = \vert\states\vert \cdot (2\cdot l\cdot \largestW + 1)^{\dimension}$, where $c$ is a constant independent of $\game$. Thus the overall memory needed to win this game is bounded by the single exponential $\depth \cdot \width$.
\end{lemma}

The sketch of this proof is the following. By Lemma \ref{lemma_depth_parity}, we know that there exists a tree $\tree$, and thus a DAG, that satisfies the bound on depth. We construct a finite sequence of DAGs, whose first element is $\tree$, so that (1) each DAG describes a winning strategy for the same initial credit, (2) each DAG has the same depth, and (3) the last DAG of the sequence has its width bounded by $\vert\states\vert \cdot (2\cdot l\cdot \largestW + 1)^{\dimension}$. This sequence $\DAG_{0} = \tree, \DAG_{1}, \DAG_{2}, \ldots{}, \DAG_{n}$ is built by merging nodes on the same level of the initial tree depending on their labels, level by level. The key idea of this procedure is that what actually matters for $\playerOne$ is only the current energy level, which is encoded in node labels in the self-covering tree $\tree$. Therefore, we merge nodes with identical states and energy levels: since $\playerOne$ can essentially play the same strategy in both nodes, we only keep one of their subtrees.

It is possible to further reduce the practical size of the compressed resulting DAG by merging nodes according to a ``greater or equal'' relation over energy levels rather than simply equality (Fig. \ref{fig:merge}). This improvement is part of the algorithm that follows, and it has a significant impact on the practical width of DAGs as it can then be bounded by the number of \textit{incomparable} labeling vectors instead of \textit{unequivalent} ones.

\begin{figure}[htb]
\begin{minipage}[b]{0.52\linewidth}
\hspace{0mm}\scalebox{0.9}{\begin{tikzpicture}[->,>=stealth',shorten >=1pt,auto,node
    distance=2.5cm,bend angle=45, scale=0.7, font=\scriptsize]
    \tikzstyle{p1}=[draw,ellipse,text centered,minimum width=26mm]
    \tikzstyle{p2}=[draw,rectangle,text centered,minimum width=22mm]
    \node[p2]  (0)  at (0, 0) {$\langle s_{0}, (0, 0)\rangle$};
    \node[p2]  (1) at (-2, -2) {$\langle s_{1}, (-1, 1)\rangle$};
    \node[p2]  (2) at (2, -2)  {$\langle s_{2}, (0, 2)\rangle$};
    \node[p1]  (3) at (-2, -4)  {$\langle s_{3}, (-1, 2)\rangle$};
    \node[p1]  (4)  at (2, -4) {$\langle s_{3}, (0, 2)\rangle$};
    \node[p1]  (5)  at (-2, -6) {$\langle s_{4}, (0, 1)\rangle$};
    \node[p1]  (6)  at (2, -6) {$\langle s_{5}, (-2, 3)\rangle$};
    \node[p2]  (7)  at (-2, -8) {$\langle s_{0}, (0, 0)\rangle$};
    \node[p1]  (8)  at (2, -8) {$\langle s_{3}, (0, 3)\rangle$};
    \coordinate[shift={(0mm,5mm)}] (init) at (0.north);
    \coordinate[shift={(2mm,3mm)}] (4a) at (4.west);
    \coordinate[shift={(24mm,3mm)}] (4b) at (4.west);
    \coordinate[shift={(24mm,-3mm)}] (8b) at (8.west);
    \coordinate[shift={(2mm,-3mm)}] (8a) at (8.west);
    \coordinate[shift={(13mm,-8mm)}] (lc1) at (2.west);
    \coordinate[shift={(-4mm,0mm)}] (lc2) at (lc1.west);
    \coordinate[shift={(0mm,4mm)}] (lc3) at (lc1.west);
    \coordinate[shift={(-4mm,4mm)}] (lc4) at (lc1.west);
    \path
    (init) edge (0);
	\draw[->,>=latex] (0) to[out=225,in=45] (1);
	\draw[->,>=latex] (0) to[out=315,in=135] (2);
	\draw[->,>=latex] (1) to[out=270,in=90] (3);
	\draw[->,>=latex] (2) to[out=270,in=90] (4);
	\draw[->,>=latex] (3) to[out=270,in=90] (5);
	\draw[->,>=latex] (4) to[out=270,in=90] (6);
	\draw[->,>=latex] (5) to[out=270,in=90] (7);
	\draw[->,>=latex] (6) to[out=270,in=90] (8);
	\draw[->,dashed,>=latex] (7) to[out=180,in=180] (0);
	\draw[->,dashed,>=latex] (8) to[out=5,in=355] (4);
	\draw[-,thick,>=latex] (4a) to (8b);
	\draw[-,thick,>=latex] (4b) to (8a);
	\draw[-,thick,>=latex] (lc1) to (lc4);
	\draw[-,thick,>=latex] (lc2) to (lc3);
	\draw[->,double,>=latex] (2) to[out=235,in=45] (3);
      \end{tikzpicture}}
      \caption{Merge between comparable nodes.}
\label{fig:merge}

\end{minipage}
\hfill
\begin{minipage}[b]{0.48\linewidth}
\centering
\scalebox{1}{\begin{tikzpicture}[->,>=stealth',shorten >=1pt,auto,node
    distance=2.5cm,bend angle=45, scale=0.35]
    \tikzstyle{p0}=[]
    \node[p0]  (0)  at (0, 0) {$\dagRoot$};
    \node[p0]  (1)  at (-1, -2) {};
    \node[p0]  (2)  at (-2, -4) {$\otherNode$};
    \node[p0]  (3)  at (-3, -6) {};
    \node[p0]  (4)  at (-4, -8) {$\nu$};
    \node[p0]  (5)  at (-6, -12) {$\node$};
    \node[p0]  (6)  at (2, -11) {$\xi$};
    \path
    ;
	\draw[->,>=latex] (0) to (1);
	\draw[->,>=latex] (1) to (2);
	\draw[->,>=latex] (2) to (3);
	\draw[->,>=latex] (3) to (4);
	\draw[->,>=latex] (4) to (5);
	\draw[->,>=latex] (4) to (6);
	\draw[dashed,->,>=latex] (6) to[out=65,in=330] (1);
	\draw[dashed,->,>=latex] (6) to[out=80,in=350] (2);
	\draw[dashed,->,>=latex] (6) to[out=90,in=0] (3);
	\draw[dashed,->,>=latex] (6) to[out=100,in=0] (4);
      \end{tikzpicture}}
      \caption{Cycles have positive energy levels.}
\label{fig:noCycle}
\end{minipage}
\end{figure}

The remainder of this subsection is dedicated to the proof of Lemma \ref{lemma_width}. We need to introduce some notations and two intermediate lemmas. If he so wishes, the reader may directly proceed to the next subsection and Lemma \ref{lemma_expFamily} for results on lower memory bounds.

We first introduce some notations. Let $\treeFull$ be a self-covering tree (i.e., epSCT without the parity condition). We define the partial order $\nodeLess$ on $\treeStates$ such that for all $\node_{1}, \node_{2} \in \treeStates$ such that $\lab(\node_{1}) = \langle t_{1}, u_{1}\rangle$ and $\lab(\node_{2}) = \langle t_{2}, u_{2}\rangle$, we have $\node_{1} \nodeLess \node_{2}$ iff $t_{1} = t_{2}$ and $u_{1} \leq u_{2}$. We denote the equivalence by $\nodeEqual$ such that $\node_{1} \nodeEqual \node_{2}$ iff $\node_{1} \nodeLess \node_{2}$ and $\node_{2} \nodeLess \node_{1}$. For all $\node \in \treeStates$, let $\ancestors$ and $\enAncestors$ resp. denote the set of \textit{ancestors} and \textit{energy ancestors} of $\node$ in $\tree$: $
\ancestors(\node) = \left\lbrace \otherNode \in \treeStates \setminus \lbrace\node\rbrace \;\vert\; \otherNode \vDash \exists\diamondsuit\node\right\rbrace$, where we use the classical {\sf CTL} notation to denote that there exists a path from $\otherNode$ to $\node$ in $\tree$, and $\enAncestors(\node) = \left\lbrace \otherNode \in \ancestors(\node) \;\vert\; \otherNode \nodeLess \node \right\rbrace$.

We build a sequence of DAGs $\dagSeq \equiv \DAG_{0} = \tree, \DAG_{1}, \DAG_{2}, \ldots{}, \DAG_{n}$ such that for all $0 < i \leq n$, $\DAG_{i}$ is obtained from $\DAG_{i-1}$ by \textit{merging} two equivalent nodes of the same minimal level (i.e., closest to the root) of $\DAG_{i-1}$. The sequence stops when we obtain a DAG $\DAG_{n} = (\dagStates_{n}, \dagEdges_{n})$ such that for all level $j$ of $\DAG_{n}$, there does not exist two distinct equivalent nodes on level $j$. This construction induces merges by increasing depth, starting with level one. Moreover, if a DAG $\DAG_{i}$ of the sequence is the result on merges up to level $j$, then it has the tree property (i.e., every node has a unique father) for levels greater than $j$. As the depth and the branching degree of $\tree$ are finite, the defined sequence of DAGs is finite (and actually bounded).

Let us give a formal definition of the \textit{merge} operation. Consider such a DAG $\DAG_{i} = (\dagStates_{i}, \dagEdges_{i})$. Let $j$ the minimal level of $\DAG_{i}$ that contains two equivalent nodes. Let $\node_{1}, \node_{2} \in \dagStates_{i}(j)$ (i.e., nodes of level $j$) be two nodes such that $\node_{1} \neq \node_{2}$ and $\node_{1} \nodeEqual \node_{2}$. We suppose w.l.o.g. an arbitrary order on nodes of the same level so that $\node_{1}, \node_{2}$ are the two leftmost nodes that satisfy this condition. We define $\DAG_{i+1}= (\dagStates_{i+1}, \dagEdges_{i+1}) = \mergeOp(\DAG_{i})$ as the result of the following transformation:
\begin{itemize}
\item $\dagStates_{i+1} = \dagStates_{i} \setminus \left(\lbrace\node_{2}\rbrace \cup \left\lbrace \node_{d} \in \dagStates_{i} \,\vert\, \node_{2} \in \ancestors(\node_{d}) \right\rbrace\right)$,
\item $\dagEdges_{i+1} = \left( \dagEdges_{i} \,\cap\, (\dagStates_{i+1} \times \dagStates_{i+1})\right) \cup \left\lbrace (\otherNode, \node_{1}) \;\vert\; (\otherNode, \node_{2}) \in \dagEdges_{i}\right\rbrace$.
\end{itemize}
Thus, we eliminate the subtree starting in $\node_{2}$ and replace all edges that point to $\node_{2}$ by edges pointing to $\node_{1}$. This follows the idea that the same strategy can be played in $\node_{2}$ as in $\node_{1}$ since the present state and the energy level are the same. 

Let $\DAG_{i} = (\dagStates_{i}, \dagEdges_{i})$ be a DAG of the sequence $\dagSeq$. Given $\node \in \dagStates_{i}$, $\otherNode \in \ancestors(\node)$, we denote by $\otherNode \ancPath \node$ an arbitrary downward path from $\otherNode$ to $\node$ in $\DAG_{i}$. Given a leaf $\node \in \dagStates_{i}$, we denote its oldest energy ancestor by $\oldestEnAnc(\node)$. Recall that a strategy is described by such a DAG according to moves of a pebble. Given a leaf $\node \in \dagStates_{i}$ and one of its energy ancestors $\otherNode \in \enAncestors(\node)$, we represent the pebble going up from $\node$ to $\otherNode$ by $\node\pebbleUp\otherNode$. Given $\pebblePathA, \pebblePathB \in (\dagStates_{i})^{\ast}$, $\pebblePathA \pebbleUp \pebblePathB$ naturally extends this notation such that we have $\last(\pebblePathA) \pebbleUp \first(\pebblePathB)$. We consider energy levels of paths in the tree by refering to their counterparts in the game. Note that given $\otherNode, \node \in \dagStates_{i}$, $\lab(\otherNode) = \langle t, u\rangle$, $\lab(\node) = \langle t', u'\rangle$, we have $\el(\otherNode \ancPath \node) = u' - u$. We start with two useful lemmas.

\begin{lemma}
\label{lemma_ancPath}
Let $\DAG_{i} = (\dagStates_{i}, \dagEdges_{i})$ be a DAG of $\dagSeq$. For all nodes $\node_{1}, \node_{2} \in \dagStates_{i}$ such that $\node_{1} \nodeEqual \node_{2}$, we have that $\forall\, \otherNode \in \ancestors(\node_{1}) \cap \ancestors(\node_{2})$, $\el(\otherNode \ancPath \node_{1}) = \el(\otherNode \ancPath \node_{2})$.
\end{lemma}

\begin{proof}
The proof is straightforward.\qed
\end{proof}

\begin{lemma}
\label{lemma_noCycle}
Let $\DAG_{i} = (\dagStates_{i}, \dagEdges_{i})$ be a DAG of $\dagSeq$. Let $\node, \otherNode, \nu, \xi \in \dagStates_{i}$ be four nodes such that $\node$ and $\xi$ are leafs, $\nu$ is the deepest common ancestor of $\node$ and $\xi$, and $\otherNode$ is an ancestor of $\nu$. Let the oldest energy ancestor of $\xi$ be an ancestor of $\node$, i.e., $\oldestEnAnc(\xi) \in \ancestors(\node)$. We have that $\el(\otherNode \ancPath \node) \leq \el(\otherNode \ancPath \nu  \ancPath \xi \pebbleUp \oldestEnAnc(\xi) \ancPath \node)$.
\end{lemma}

This lemma states that we can extract pebble cycles, which have positive energy levels, from a given path, in order to obtain some canonical path whose energy level is lower or equal (Fig. \ref{fig:noCycle}).

\begin{proof}
Let $\nodeNC = \oldestEnAnc(\xi)$ and $\prefix = \otherNode \ancPath \nu \ancPath \xi\pebbleUp\nodeNC \ancPath \node$. Since $\nodeNC \in \ancestors(\node) \cap \ancestors(\xi)$, we have $\nodeNC \in \ancestors(\nu) \cup \lbrace\nu\rbrace$. Therefore, and applying Lemma \ref{lemma_ancPath}, four cases are possible: $\nodeNC \in \ancestors(\otherNode)$, $\nodeNC = \otherNode$, $\nodeNC \in \ancestors(\nu) \setminus \left(\ancestors(\otherNode) \cup \lbrace\otherNode\rbrace\right)$, and $\nodeNC = \nu$. Consider the first case, $\nodeNC \in \ancestors(\otherNode)$. Then $\prefix = \otherNode \ancPath \nu \ancPath \xi \pebbleUp \nodeNC \ancPath \otherNode \ancPath \nu \ancPath \node$. We have $\el(\prefix) = \el(\otherNode \ancPath \nu) + \el(\nu \ancPath \xi) + \el(\nodeNC \ancPath \otherNode) + \el(\otherNode \ancPath \nu) + \el(\nu \ancPath \node) = \el(\nodeNC \ancPath \otherNode \ancPath \nu \ancPath \xi) + \el(\otherNode \ancPath \node)$. By definition of $\nodeNC = \oldestEnAnc(\xi)$, the first term is positive. Thus, $\el(\prefix) \geq \el(\otherNode \ancPath \node)$. Arguments are similar for the other cases.\qed
\end{proof}

We proceed with the proof of Lemma \ref{lemma_width}.
\begin{proof}[Lemma \ref{lemma_width}]
Let $\dagSeq$ be the sequence of DAGs defined above. We claim that \textbf{(i)} each DAG describes a winning strategy for the same initial credit, \textbf{(ii)} each DAG has the same depth $l$, and \textbf{(iii)} the last DAG of the sequence has its width bounded by $\vert\states\vert \cdot (2\cdot l\cdot \largestW + 1)^{\dimension}$.

\textbf{(i)} First, recall that $\playerOne$ can play a strategy $\strat_{1}^{\tree} \in \stratsPureFinite_{1}$ based on edges taken by a pebble on $\tree$. Notice that moving the pebble as we previously defined is possible because nodes belonging to $\playerOne$ have only one child, and nodes of $\playerTwo$ have childs covering all his choices once, and only once. Fortunately, the $\mergeOp$ operation maintains this property. Therefore, it is straightforward to see that $\playerOne$ can also play a strategy $\strat_{1}^{\DAG_{i}} \in \stratsPureFinite_{1}$ for a DAG $\DAG_{i}$ resulting of some merges on $\tree$. However, while this would be a valid strategy for $\playerOne$, we have to prove that it is still a winning one, for the same initial credit $\initCredit$ as $\strat_{1}^{\tree}$. Precisely, we claim that $\forall\, i \geq 0$, we have that $\strat_{1}^{\DAG_{i}}$ is winning for $\initCredit$.

We show it by induction on $\DAG_{i}$. The base case is trivial as $\DAG_{0} = \tree$: the strategy $\strat_{1}^{\tree}$ is winning for $\initCredit$ by definition. Our induction hypothesis is that our claim is valid for $\DAG_{i-1}$, and we now prove it for $\DAG_{i}$, by contradiction. Let $\node_{1}, \node_{2} \in \dagStates_{i-1}(j)$ be the merged nodes, for some level $j$ of $\DAG_{i-1}$. Suppose $\strat_{1}^{\DAG_{i}}$ is not winning for $\initCredit$. Thus there exists a finite path $\pebblePath$ of the pebble in $\DAG_{i}$, which corresponds to a strategy $\strat_{2}^{\DAG_{i}} \in \stratsPureFinite_{2}$ of $\playerTwo$, such that it achieves a negative value on at least one dimension $m$, $1 \leq m \leq \dimension$. We have that $\left(\initCredit + \el(\pebblePath)\right)(m) < 0$. We aim to find a similar path $\pebblePathNew$ in $\DAG_{i-1}$ such that $\el(\pebblePathNew) \leq \el(\pebblePath)$, thus yielding contradiction, as it would witness that $\strat_{1}^{\DAG_{i-1}}$ is not winning for $\initCredit$.

We denote by $\nodeMerge$ the father of $\node_{2}$ in $\DAG_{i-1}$. The only edge added by the $\mergeOp$ operation is $(\nodeMerge, \node_{1})$. Obviously, if $\pebblePath$ does not involve this edge, then we can take $\eta = \pebblePath$ and im\-me\-dia\-tely obtain contradiction. Thus, we can decompose the witness path 
\begin{equation*}
\pebblePath = \pebblePathA(1)\,\nodeMerge\node_{1}\,\pebblePathB(1)\pebbleUp\pebblePathA(2)\,\nodeMerge\node_{1}\,\pebblePathB(2)\pebbleUp\ldots{}\pebbleUp\pebblePathA(q)\,\nodeMerge\node_{1}\,\pebblePathE,
\end{equation*}
for some $q \geq 1$ such that for all $1 \leq p \leq q$, we have that $\pebblePathA(p), \pebblePathB(p), \pebblePathE \in \left( \dagStates_{i} \cup \lbrace\pebbleUp\rbrace\right)^{\ast}$ are valid paths of the pebble in $\DAG_{i}$ (and $\DAG_{i-1}$); they do not involve edge $(\nodeMerge, \node_{1})$, i.e., $\lbrace\nodeMerge\node_{1}\rbrace \not\subseteq \pebblePathA(p), \pebblePathB(p), \pebblePathE$; and $\pebblePathB(p) \cap \left( \ancestors_{\DAG_{i}}(\nodeMerge) \setminus \ancestors_{\DAG_{i-1}}(\node_{1})\right) = \emptyset$, $\last(\pebblePathB(p))$ is a leaf and $\oldestEnAnc(\last(\pebblePathB(p))) \in \ancestors_{\DAG_{i}}(\nodeMerge)$.

Intuitively, $\pebblePath$ is split into several parts in regard to $q$, the number of times it takes the added edge $(\nodeMerge, \node_{1})$. Each time, this transition is preceded by some path $\pebblePathA$. It is then followed by some path $\pebblePathB$ where all visited ancestors of $\nodeMerge$ were already ancestors of $\node_{1}$ in $\DAG_{i-1}$ (thus, $\pebblePathB$ paths can be kept in $\pebblePathNew$). Finally, after the $q$-th transition $\nodeMerge\node_{1}$ is taken, the path $\pebblePath$ ends with a finite sub-path $\pebblePathE$.

We define the witness path $\pebblePathNew$ in $\DAG_{i-1}$ as $\pebblePathNew = \pebblePathNewA(1)\pebblePathB(1)\pebbleUp\pebblePathNewA(2)\pebblePathB(2)\pebbleUp\ldots{}\pebbleUp\pebblePathNewA(q)\pebblePathE,$
with the following transformation of sub-paths $\pebblePathA(p)\,\nodeMerge\node_{1}$:\begin{itemize}
\item $\pebblePathNewA(1) = \dagRoot \ancPath_{\DAG_{i-1}} \node_{1}$,
\item $\forall\, 2 \leq p \leq q, \pebblePathNewA(p) = \oldestEnAnc(\last(\pebblePathB(p-1))) \ancPath_{\DAG_{i-1}} \node_{1}$,
\end{itemize}
where $\ancPath_{\DAG_{i-1}}$ denotes a valid path in $\DAG_{i-1}$. Note that given preceding definitions, this indeed constitutes a valid path in $\DAG_{i-1}$. We have to prove that $\el(\pebblePathNew) \leq \el(\pebblePath)$.
We have 
\begin{equation*}
\el(\pebblePathNew) = \sum_{1 \leq p \leq q} \el(\pebblePathNewA(p)) + \sum_{1 \leq p \leq q-1} \el(\pebblePathB(p)) + \el(\pebblePathE),
\end{equation*}
and
\begin{equation*}
\el(\pebblePath) = \sum_{1 \leq p \leq q} \el(\pebblePathA(p)\,\nodeMerge\node_{1}) + \sum_{1 \leq p \leq q-1} \el(\pebblePathB(p)) + \el(\pebblePathE).
\end{equation*}
Thus, it remains to show that
\begin{equation*}
\sum_{1 \leq p \leq q} \el(\pebblePathNewA(p)) \leq \sum_{1 \leq p \leq q} \el(\pebblePathA(p)\,\nodeMerge\node_{1}).
\end{equation*}

In particular, we claim that for all $1 \leq p \leq q$, we have $\el(\pebblePathNewA(p)) \leq \el(\pebblePathA(p)\,\nodeMerge\node_{1})$. Indeed, notice that $\pebblePathNewA(p)$ and $\pebblePathA(p)$ share their starting and ending nodes and that $\pebblePathA(p)$ contains a finite number of pebble cycles. Let $\otherNode$ denote the common starting node of both $\pebblePathNewA(p)$ and $\pebblePathA(p)$. Applying Lemma \ref{lemma_noCycle} on $\pebblePathA(p)$, we can eliminate cycles one at a time, without ever increasing the energy level, and obtain a path $\otherNode \ancPath_{\DAG_{i}} \nodeMerge\node_{1}$ such that $\el(\otherNode \ancPath_{\DAG_{i}} \nodeMerge\node_{1}) \leq \el(\pebblePathA(p)\,\nodeMerge\node_{1})$. Since $\node_{1} \nodeEqual \node_{2}$, we have by Lemma \ref{lemma_ancPath} that $\el(\otherNode \ancPath_{\DAG_{i}} \nodeMerge\node_{1}) = \el(\otherNode \ancPath_{\DAG_{i-1}} \nodeMerge\node_{2}) = \el(\otherNode \ancPath_{\DAG_{i-1}} \node_{1})$, implying the claim.

Consequently, we obtain $\el(\pebblePathNew) \leq \el(\pebblePath)$, which witnesses that $\DAG_{i-1}$ was not winning. This contradicts our induction hypothesis and concludes our proof that for all $0 \leq i \leq n$, $\strat_{1}^{\DAG_{i}}$ is winning for $\initCredit$.

\textbf{(ii)} Second, the $\mergeOp$ operation only prunes some parts of the tree $\tree$, without ever adding any new state, and added edges are on existing successive levels. Therefore, each $\DAG_{i}$ has noticeably the same depth $l$.

\textbf{(iii)} Third, the last DAG of the sequence, $\DAG_{n}$, is such that for all level $j$, for all $\node_{1}, \node_{2} \in \dagStates_{n}(j)$, we have $(\node_{1} \neq \node_{2}) \Rightarrow (\node_{1} \not\nodeEqual \node_{2})$. Therefore the width of this DAG is bounded by the number of possible non-equivalent nodes. Recall that two nodes are equivalent if they have the same labels, i.e., they represent the same state of the game and are marked with exactly the same energy level vector. Since the maximal change in energy level on an edge is $\largestW$, and the depth of the DAG is bounded by $l = 2^{(d-1) \cdot \vert\states\vert} \cdot \left( \largestW \cdot \vert\states\vert + 1\right) ^{c \cdot \dimension^{2}}$ thanks to Lemma \ref{lemma_depth_parity}, we have possible vectors in $\lbrace -l\cdot\largestW, -l\cdot\largestW + 1, \ldots{}, l\cdot\largestW - 1, l\cdot\largestW\rbrace^{\dimension}$ for each state. Consequently, the width of $\DAG_{n}$ is bounded by
\begin{equation*}
\vert\states\vert \cdot (2\cdot l\cdot \largestW + 1)^{\dimension} = 
\vert\states\vert \cdot \left( 2^{d \cdot \vert\states\vert} \cdot \left( \largestW \cdot \vert\states\vert +1\right) ^{c \cdot \dimension^{2}}\cdot \largestW + 1\right) ^{\dimension},
\end{equation*}
which is still single exponential.\qed
\end{proof}

\paragraph{{\bf Lower bound.}} In the next lemma, we show that the upper bound is tight in the sense that there exist families of games which require exponential memory (in the number of dimensions), even for the simpler case of multi energy objectives without parity and weights in $\lbrace -1, 0, 1\rbrace$ (Fig. \ref{fig:expFamily}). Note that for one-dimension energy parity, it was shown in \cite{chatterjee_ICALP10} that exponential memory (in the encoding of weights) may be necessary.

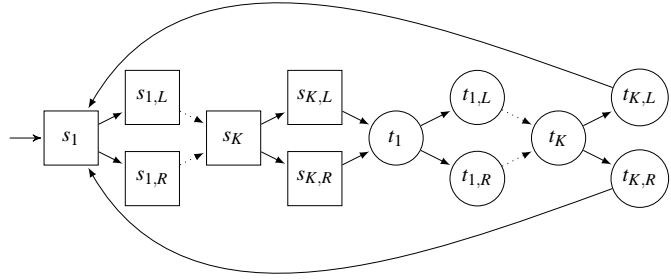
\begin{figure}[htb]
  \centering   
  \scalebox{0.9}{\begin{tikzpicture}[->,>=stealth',shorten >=1pt,auto,node
    distance=2.5cm,bend angle=45,scale=0.3, font=\small]
    \tikzstyle{p1}=[draw,circle,text centered,minimum size=8mm]
    \tikzstyle{p2}=[draw,rectangle,text centered,minimum size=8mm]
    \node[p2]  (0)  at (0, 0) {$s_{1}$};
    \node[p2]  (1) at (4, 2) {$s_{1,L}$};
    \node[p2]  (2) at (4, -2)  {$s_{1,R}$};
    \node[p2]  (3) at (8, 0)  {$s_{\gadgets}$};
    \node[p2]  (4)  at (12, 2) {$s_{\gadgets,L}$};
    \node[p2]  (5)  at (12, -2) {$s_{\gadgets,R}$};
    \node[p1]  (6)  at (16, 0) {$t_{1}$};
    \node[p1]  (7) at (20, 2) {$t_{1,L}$};
    \node[p1]  (8) at (20, -2)  {$t_{1,R}$};
    \node[p1]  (9) at (24, 0)  {$t_{\gadgets}$};
    \node[p1]  (10)  at (28, 2) {$t_{\gadgets,L}$};
    \node[p1]  (11)  at (28, -2) {$t_{\gadgets,R}$};
    \coordinate[shift={(-5mm,0mm)}] (init) at (0.west);
    \path
    (init) edge (0);
	\draw[->,>=latex] (0) to (1);
	\draw[->,>=latex] (0) to (2);
	\draw[->,>=latex] (3) to (4);
	\draw[->,>=latex] (3) to (5);
	\draw[dotted,->,>=latex] (1) to (3);
	\draw[dotted,->,>=latex] (2) to (3);
	\draw[->,>=latex] (4) to (6);
	\draw[->,>=latex] (5) to (6);
	\draw[->,>=latex] (6) to (7);
	\draw[->,>=latex] (6) to (8);
	\draw[dotted,->,>=latex] (7) to (9);
	\draw[dotted,->,>=latex] (8) to (9);
	\draw[->,>=latex] (9) to (10);
	\draw[->,>=latex] (9) to (11);
	\draw[->,>=latex] (10) to[out=160,in=60] (0);
	\draw[->,>=latex] (11) to[out=200,in=300] (0);
      \end{tikzpicture}}
\vspace*{-8mm}
      \caption{Family of games requiring exponential memory.}
\label{fig:expFamily}
  \end{figure}

\begin{lemma}
\label{lemma_expFamily}
There exists a family of multi energy games $(\game(\gadgets))_{\gadgets \geq 1} = \left( \statesOne, \statesTwo, \initState, \edges, \right. $ $\left. \dimension = 2\cdot\gadgets, \weight: \edges \rightarrow \lbrace -1, 0, 1\rbrace^{\dimension}\right)$ such that for any initial credit, $\playerOne$ needs exponential memory to win.
\end{lemma}

The idea is the following: in the example of Fig. \ref{fig:expFamily}, if $\playerOne$ does not remember the exact choices of $\playerTwo$ (which requires an exponential size Moore machine), there will exist some sequence of choices of $\playerTwo$ such that $\playerOne$ cannot counteract a decrease in energy. Thus, by playing this sequence long enough, $\playerTwo$ can force $\playerOne$ to lose, whatever his initial credit is.

\begin{proof}
We define a family of games $(\game(\gadgets))_{\gadgets \geq 1}$ which is an assembly of $\dimension = 2\cdot\gadgets$ gadgets, the first $\gadgets$ belonging to $\playerTwo$, and the remaining $\gadgets$ belonging to $\playerOne$ (Fig. \ref{fig:expFamily}). Precisely, we have $\vert\states_{1}\vert = \vert\states_{2}\vert = 3\cdot\gadgets$, $\vert\states\vert = \vert\edges\vert = 6\cdot\gadgets = 3\cdot \dimension$ (linear in $\dimension$), $\dimension = 2\cdot\gadgets$, and $\weight$ defined as:
\begin{align*}
\forall\, 1 \leq i \leq \gadgets,\, &\weight((\circ, s_{i})) = \weight((\circ, t_{i})) = (0, \ldots{}, 0),\\
&\weight((s_{i}, s_{i, L})) = - \weight((s_{i}, s_{i, R})) = \weight((t_{i}, t_{i, L})) = - \weight((t_{i}, t_{i, R})),\\
&\forall\, 1 \leq j \leq \dimension,\, \weight((s_{i}, s_{i, L}))(j) = \begin{cases}1 \text{ if } j = 2\cdot i - 1\\-1 \text{ if } j = 2\cdot i\\0 \text{ otherwise}\end{cases},
\end{align*}
where $\circ$ denotes any valid predecessor state.

There exists a winning strategy $\strat^{exp}_{1}$ for $\playerOne$, for initial credit $v^{exp}_{0} = (1, \ldots{}, 1)$. Indeed, for any strategy of $\playerTwo$, for any state $t_{i}$ belonging to $\playerOne$, it suffices to play the \textit{opposite} choice as $\playerTwo$ made on its last visit of $s_{i}$ to maintain at all times an energy vector which is positive on all dimensions. This strategy thus requires to remember the last choice of $\playerTwo$ in all gadgets, which means $\playerOne$ needs $\gadgets$ bits to encode these decisions. Thus, this winning strategy is described by a Moore machine containing $2^{\gadgets} = 2^{\frac{\dimension}{2}}$ states, which is exponential in the number of dimensions $\dimension$.

We claim that, for any initial credit $v_{0}$, there exists no winning strategy $\strat_{1}$ that can be described with less than $2^{\gadgets}$ states and prove it by contradiction. Suppose $\playerOne$ plays according to such a strategy $\strat_{1}$. Then there exists some $1 \leq x \leq \gadgets$ such that
$\strat_{1}(s_{1} \ldots{} s_{x} s_{x, L} \ldots{} t_{x}) =  \strat_{1}(s_{1} \ldots{} s_{x} s_{x, D} \ldots{} t_{x})$, i.e., $\playerOne$ chooses the same action in $t_{x}$ against both choices of the adversary. Suppose that $\playerOne$ chooses to play $t_{x, L}$ in both cases, that is $\strat_{1}(s_{1} \ldots{} s_{x} s_{x, L} \ldots{} t_{x}) =  \strat_{1}(s_{1} \ldots{} s_{x} s_{x, D} \ldots{} t_{x}) = t_{x, L}$. By playing $s_{x, L}$, $\playerTwo$ can force a decrease of the energy vector by $2$ on dimension $2\cdot x$ every visit in gadget $x$. Similarly, if the strategy of $\playerOne$ is to play $t_{x, R}$, $\playerTwo$ wins by choosing to play $s_{x, R}$ as dimension $2\cdot x - 1$ decreases by $2$ every visit. Therefore, whatever the finite initial vector of $\playerOne$, $\playerTwo$ can enforce a negative dimension by playing long enough. This contradicts the fact that $\strat_{1}$ is winning and concludes our proof that exponential memory is necessary for this simple family of games $(\game(\gadgets))_{\gadgets \geq 1}$.\qed
\end{proof}

We summarize our results in Theorem \ref{thm_optBounds}.

\begin{theorem}[Optimal memory bounds]
\label{thm_optBounds}
The following assertions hold: (1) In multi energy parity games, if there exists a winning strategy, then there exists a finite-memory winning strategy. (2) In multi energy parity and multi mean-payoff games, if there exists a finite-memory winning strategy, then there exists a winning strategy with at most exponential memory. (3) There exists a family of multi energy games (without parity) with weights in $\lbrace -1, 0 ,1\rbrace$ where all winning strategies require at least exponential memory.
\end{theorem}

\begin{proof}
Thanks to \cite[Theorem 3]{chatterjee_FSTTCS10}, we have equivalence between finite-memory winning for multi energy and multi mean-payoff games. The rest follows from straigthforward application of Lemma \ref{purefinitememory}, Lemma \ref{lemma_parityToEnergy}, Lemma \ref{lemma_width}, and Lemma~\ref{lemma_expFamily}.\qed
\end{proof}

\section{Symbolic synthesis algorithm}
\label{sec:alg}

We now present a {\em symbolic}, {\em incremental} and {\em optimal} algorithm to synthesize a finite-memory winning strategy in a MEG.\footnote{Note that the symbolic algorithm can be applied to MEPGs and MMPPGs after removal of the parity condition by applying the construction of Lemma~\ref{lemma_parityToEnergy}.} This algorithm outputs a (set of) winning initial credit(s) and a derived finite-memory winning strategy (if one exists) which is exponential in the worst-case. Its running time is at most exponential. So our symbolic algorithm can be considered (worst-case) optimal in the light of the results of previous section. 

This algorithm computes the greatest fixed point of a monotone operator that defines the sets of winning initial (vectors of) credits for each state of the game. As those sets are upward-closed, they are symbolically represented by their minimal elements. To ensure convergence, the algorithm considers only credits that are below some {\em threshold}, noted $\maxE$. This is without giving up completeness because, as we show below, for a game $\gameFull$, it is sufficient to take the value $2\cdot l \cdot W$ for $\maxE$, where $l$ is the bound on the depth of epSCTs obtained in Lemma~\ref{lemma_depth_parity} and $W$ is the largest absolute value of weights used in the game. We also show how to extract a finite state Moore machine representing a corresponding winning strategy (states of the Moore machine encode the memory of the strategy) from this set of minimal winning initial credits and how to obtain an {\em incremental} algorithm by increasing values for the threshold $\maxE$ starting from small values. 

\paragraph{{\bf A controllable predecessor operator.}}
Let $\gameFull$ be a MEG, $\maxE \in \nat$ be a constant, and $U(\maxE)$ be the set $(\statesOne \cup \statesTwo) \times \lbrace 0, 1,\ldots{},\maxE\rbrace^{\dimension}$. Let ${\cal U}(\maxE)=2^{U(\maxE)}$, i.e., the powerset of $U(\maxE)$, and the operator $\Cpre_{\maxE} \colon  {\cal U}(\maxE) \rightarrow {\cal U}(\maxE)$ be defined as follows:
\begin{align*}
\mathcal{E}(V) &= \{ (s_1,e_1) \in U(\maxE) \mid s_1 \in \statesOne \land \exists (s_1,s) \in E, \exists (s,e_2) \in V : e_2 \leq e_1 + w(s_1,s) \},\\
\mathcal{A}(V) &= \{ (s_2,e_2) \in U(\maxE) \mid s_2 \in \statesTwo \land \forall (s_2,s) \in E, \exists (s,e_1) \in V : e_1 \leq e_2 + w(s_2,s) \},
\end{align*}
\begin{equation}
\label{eq:cpre}
\Cpre_{\maxE}(V) = \mathcal{E}(V) \;\cup\; \mathcal{A}(V).
\end{equation}
Intuitively, $\Cpre_{\maxE}(V)$ returns the set of energy levels from which \playerOne can force an energy level in $V$ in one step. 
The operator $\Cpre_{\maxE}$ is $\subseteq$-monotone over the complete lattice ${\cal U}(\maxE)$, and so there exists a \textit{greatest fixed point} for $\Cpre_{\maxE}$ in the lattice ${\cal U}(\maxE)$, denoted by $\Cpre_{\maxE}^{\ast}$. As usual, the greatest fixed point of the operator  $\Cpre_{\maxE}$ can be computed by successive ap\-proxi\-ma\-tions as the last element of the following finite $\subseteq$-descending chain. We define the algo\-rithm $\algoName$ that computes this greatest fixed point:
\begin{equation}
\label{eq:desc-chain}
U_0=U(\maxE),\;U_1=\Cpre_{\maxE}(U_0),\;\ldots,\;U_n=\Cpre_{\maxE}(U_{n-1})=U_{n-1}.
\end{equation}
The set $U_i$ contains all the energy levels that are sufficient to maintain the energy positive in all dimensions for $i$ steps. Note that the length of this chain can be bounded by $| U(\maxE) |$ and the time needed to compute each element of the chain can be bounded by a polynomial in $| U(\maxE) |$. As a consequence, we obtain the following lemma.

\begin{lemma}
\label{lemma_cpre}
Let $\gameFull$ be a multi energy game and $\maxE \in \nat$ be a constant. Then $\Cpre_{\maxE}^{\ast}$ can be computed in time bounded by a polynomial in $| U(\maxE) |$, i.e., an exponential in the size of $\game$.
\end{lemma}

\paragraph{{\bf Symbolic representation.}}
To define a symbolic representation of the sets manipulated by the $\Cpre_{\maxE}$ operator, we exploit the following partial order: let $(s,e), (s',e') \in U(\maxE)$, we define 
\begin{equation}
(s,e) \nodeLess (s',e')  \mbox{~iff~} s = s' \mbox{~and~} e \leq e'.
\end{equation}
\noindent
A set $V \in {\cal U}(\maxE)$ is {\em closed} if for all $(s,e), (s',e') \in U(\maxE)$, if $(s,e) \in V$ and $(s,e) \preceq (s',e')$, then $(s',e') \in V$. By definition of $\Cpre_{\maxE}$, we get the following property.

\begin{lemma}
\label{lemma_closed}
All sets $U_i$ in Eq.~\eqref{eq:desc-chain} are closed for $\preceq$. 
\end{lemma}

Therefore, all sets $U_{i}$ in the descending chain of Eq.~\eqref{eq:desc-chain} can be symbolically represented by their minimal elements $\minElems(U_{i})$ which is an antichain of elements for $\preceq$. Even if the largest antichain can be exponential in $\game$, this representation is, in practice, often much more efficient, even for small values of the parameters. For example, with $\maxE = 4$ and $\dimension = 4$, we have that the cardinality of a set can be as large as $\vert U_{i}\vert \leq 625$ whereas the size of the largest antichain is bounded by $\vert \minElems(U_{i})\vert \leq 35$. Antichains have proved to be very efficient: see for example \cite{ACHMV10,WDHR06,DR10}. Therefore, our algorithm is expected to have good performance in practice.

\paragraph{{\bf Correctness and completeness.}}
The following two lemmas relate the greatest fixed point $\Cpre_{\maxE}^{\ast}$ and the existence of winning strategies for $\playerOne$ in $\game$. We start with the correctness of the symbolic algorithm.

\begin{lemma}[Correctness]
\label{lem:MooreFP}
Let $\gameFull$ be a multi energy game, let $\maxE \in \nat$ be a constant. If there exists $(c_{1}, \ldots{}, c_{\dimension}) \in \nat^{\dimension}$ such that $(\initState, (c_{1}, \ldots{}, c_{\dimension})) \in \Cpre_{\maxE}^{\ast}$, then $\playerOne$ has a winning strategy in $\game$ for initial credit $(c_{1}, \ldots{}, c_{\dimension})$ and the memory needed by $\playerOne$ can be bounded by $| {\sf Min}_{\preceq}(\Cpre_{\maxE}^{\ast})|$ (the size of the antichain of minimal elements in the fixed point).
\end{lemma}

Given the set of winning initial credits output by $\algoName$, it is straightforward to derive a corresponding winning strategy of at most exponential size. Indeed, for winning initial credit $\overline{c} \in \nat^{\dimension}$, we build a Moore machine which (i) states are the minimal elements of the fixed point (antichain at most exponential in $\game$), (ii) initial state is any element $(t, u)$ among them such that $t=\initState$ and $u \leq \overline{c}$, (iii) next-action function prescribes an action that ensures remaining in the fixed point, and (iv) update function maintains an accurate energy level in the memory.

\newcommand{\mooreMachineNoParam}{\ensuremath{\mathcal{M}} }
\newcommand{\mooreMachine}[1]{\ensuremath{\mathcal{M}(#1)} }
\newcommand{\mooreMachineFull}[1]{\ensuremath{\mooreMachine{#1} = (\mooreMem, \mooreInitMem, \mooreUpd, \mooreNext)} }
\newcommand{\mooreMem}{\ensuremath{M} }
\newcommand{\mooreMemElem}{\ensuremath{{\sf m}} }
\newcommand{\mooreInitMem}{\ensuremath{\mooreMemElem_{0}} }
\newcommand{\mooreUpd}{\ensuremath{\alpha_{{\sf u}}} }
\newcommand{\mooreUpdHat}{\ensuremath{\hat{\alpha}_{{\sf u}}} }
\newcommand{\mooreNext}{\ensuremath{\alpha_{{\sf n}}} }

\begin{proof}
We denote by $\overline{c}$ the $\dimension$-dimension credit vector $(c_{1}, \ldots{}, c_{\dimension})$. W.l.o.g. we assume that states of $\game$ alternate between positions of $\playerOne$ and positions of $\playerTwo$ (otherwise, we split needed edges by introducing dummy states). From $\Cpre_{\maxE}^{\ast}$, we construct a Moore machine $\mooreMachineNoParam =(\mooreMem,\mooreInitMem,\mooreUpd, \mooreNext)$ which respects the following definitions:
  \begin{itemize}
  	\item $\mooreMem ={\sf Min}_{\preceq} \{ (t,u) \in S_1 \times \lbrace 0 \ldots{} \maxE\rbrace^{\dimension} \mid (t,u) \in (\Cpre_{\maxE}^{\ast}) \}$. The set of states of the machine is the antichain of $\preceq$-minimal elements that belong to $\playerOne$ in the fixed point. Note that the length of this antichain is bounded by an exponential in the size of the game.
	\item $\mooreInitMem$ is any element $(t,u)$ in $\mooreMem$ such that $t=\initState$ and $u \leq \overline{c}$. Note that such an element is guaranteed to exist as $(\initState, \overline{c}) \in \Cpre_{\maxE}^{\ast}$.
	\item For all $(t,u) \in \mooreMem$, we define $\mooreNext((t,u))$ by choosing any element $(t,t') \in E$ such that there exists $(t',u') \in \Cpre_{\maxE}^*$ with $u'=u+w(t,t')$. Such an element is guaranteed to exist by definition of $\Cpre_{\maxE}$ and the fact that  $(t,u) \in \Cpre_{\maxE}^{\ast}$.
	\item $\mooreUpd \colon \mooreMem \times ((\statesTwo \times \states) \cap \edges) \rightarrow \mooreMem$ is any partial function that respects the following constraint: if $\mooreNext((t,u))=(t,t')$ then $\mooreUpd((t,u), (t',t''))$ is defined for any $(t',t'') \in E$ and can be chosen to be equal to any $(t'',u'')$ such that $u'' \leq u+w(t,t')+w(t',t'')$, and such an $u''$ is guaranteed to exist by definition of $\Cpre_{\maxE}$ and because $\Cpre_{\maxE}^{\ast}$ is a fixed point.
  \end{itemize}
Now, let us prove that for any initial prefix $s_0s_1  \dots s_{2n}$ of even length in $G$, which is compatible with $\mooreMachineNoParam$, we have that $\overline{c} + \el(s_0s_1 \dots s_{2n-1}) \geq 0$ and that $\overline{c} + \el(s_0s_1 \dots s_{2n}) \geq 0$. To establish this property, we first prove the following property by induction on $n$: $\overline{c} + \el(s_0s_1  \dots s_{2n}) \geq u$ where $u$ is the energy level of the label of the state reached after reading the prefix $s_0s_1  \dots s_{2n}$ with the Moore machine $\mooreMachineNoParam$. Base case $n=0$ is trivial. Induction: assume that the property is true for $n-1$, and let us establish it for $n$. By induction hypothesis, we have that $\overline{c} + \el(s_0s_1  \dots s_{2(n-1)}) \geq u$ where $u$ is the energy level of the label of state $\mooreMemElem$ that is reached after reading $s_0s_1  \dots s_{2(n-1)}$ with the Moore machine. Now, assume that $\mooreNext(\mooreMemElem)=(t,t')$. So, $s_{2(n-1)}=t$ and the choice of $\playerOne$ is to play $(t,t')$. So, $s_{2(n-1)+1}=t'$. Now for all possible choices $(t',t'')$ of $\playerTwo$, we know by definition of $\mooreMachineNoParam$ that the energy level $u''$  that labels the state $\mooreUpd(\mooreMemElem,(t',t''))$ is $u'' \leq u+w(t,t')+w(t',t'')$, which establishes our property. Therefore, the strategy of $\playerOne$ based on $\mooreMachineNoParam$ is such that the energy always stays positive for initial credit $\overline{c}$, which concludes the proof.\qed
\end{proof}

Completeness of the symbolic algorithm is guaranteed when a sufficiently large threshold $\maxE$ is used as established in the following lemma.

\begin{lemma}[Completeness]
\label{lem:gfpcpre}
Let $\gameFull$ be a multi energy game in which all absolute values of weights are bounded by $W$. If $\playerOne$ has a winning strategy in $\game$ and $T=(Q,R)$ is a self-covering tree for $\game$ of depth $l$, then $(\initState, (\maxE, \ldots{}, \maxE)) \in \Cpre_{\maxE}^{\ast}$ for $\maxE = 2\cdot l \cdot \largestW$. 
\end{lemma}

\begin{remark}
This algorithm is complete in the sense that if a winning strategy exists for $\playerOne$, it outputs at least a winning initial credit (and the derived strategy) for $\maxE = 2\cdot l \cdot \largestW$. However, this is different from the \textit{fixed initial credit problem}, which consists in deciding if a particular given credit vector is winning and is known to be EXPSPACE-hard by equivalence with deciding the existence of an infinite run in a Petri net given an initial marking \cite{brazdil_ICALP10,fahrenberg_ICTAC11}. In general, there may exist winning credits incomparable to those captured by algorithm $\algoName$. More precisely, given a constant $\maxE \in \nat$, the algorithm fully captures all the winning initial credits smaller than $(\maxE, \ldots{}, \maxE)$. Indeed, the fixed point computation considers the whole range of initial credits up to the given constant exhaustively, and only removes credits if they do not suffice to win. By Lemma \ref{lem:gfpcpre}, it is moreover guaranteed that if an arbitrary winning initial credit exists, then there exists one in the range defined by the constant $\maxE = 2\cdot l \cdot \largestW$. Nevertheless, since our algorithm works in exponential time while the problem of finding \textit{all} the winning initial credits is EXPSPACE-hard, there may be some incomparable credits outside that range that are not captured by the algorithm (comparable credits are captured since we work with upper closed sets). Indeed, if our algorithm was able to compute exhaustively all winning credits in exponential time, this would induce that EXPTIME is equal to EXPSPACE. Notice that defining a class of games for which the algorithm $\algoName$ proves to be incomplete (in the sense that incomparable winning credits exist outside the region captured by constant $\maxE = 2\cdot l \cdot \largestW$) is an interesting open problem.
\end{remark}

\begin{proof}
To establish this property, we first prove that from the set of labels of $T$, we can construct a set $f$ which is increasing for the operator $\Cpre_{\maxE}$, i.e., $\Cpre_{\maxE}(f) \supseteq f$, and such that $(\initState, (\maxE, \ldots{}, \maxE)) \in f$. We define $f$ from $T=(Q,R)$ as follows. Let $C \in \nat$ be the smallest non-negative integer such that for all $q \in Q$, with $\lab(q)=(t,u)$, for all dimensions $i$, $1 \leq i \leq \dimension$, we have that $u(i)+C \geq 0$. Integer $C$ is bounded from above by $l \cdot \largestW$ because on every path from the root to a leaf in $T$, every dimension is at most decreased $l$ times by an amount bounded by $W$, and at the root all the dimensions are equal to $0$. For any $q \in Q$, we denote by $\lab(q)+C$ the label of $q$ where the energy level has been increased by $C$ in all the dimensions, i.e., if $\lab(q)=(t,u)$ then $\lab(q)+C=(t,u+ (C,\ldots,C))$. Note that for all nodes in $Q$, the label is at most $l \cdot W$ and thus the shifted label remains under $\maxE = 2\cdot l\cdot W$. Now, we define the set $f$ as follows:
\begin{equation}
f=\{ (t,u) \in U(\maxE) \mid \exists\, q \in Q,\, \lab(q)+C \preceq (t,u) \}.
\end{equation}
\noindent
So, $f$ is defined as the $\preceq$-closure of the set of labels in $T$ shifted by $C$ in all the dimensions. 

First, note that $(\initState, (\maxE,\ldots,\maxE)) \in f$ as the label of the root in $T$ is $(\initState, (0,\ldots,0))$. Second, let us show that $\Cpre_{\maxE}(f) \supseteq f$. Take any $(t,u) \in f$ and let us show that $(t,u) \in \Cpre_{\maxE}(f)$. We decompose the proof in two cases. \textbf{(A)} $t \in \statesOne$. By definition of $f$, there exists $q \in Q$ such that $\lab(q)+C \preceq (t,u)$. W.l.o.g. we can assume that $q$ is not a leaf as otherwise there exists an ancestor $q'$ of $q$ such that $\lab(q') \preceq \lab(q)$ (recall the set is described by its minimal elements). By definition of $T$, there exists $(t,t') \in E$ and $q' \in Q$ such that $(q,q') \in R$ and $\lab(q') = \lab(q) + \weight(t, t')$. Let $(t', v) =  \lab(q') + C$.
By definition of $f$, we have $(t',v) \in f$. By Eq.~\eqref{eq:cpre}, it follows that $(t,u) \in \Cpre_{\maxE}(f)$. \textbf{(B)} $t \in \statesTwo$. By definition of $f$, there exists $q \in Q$ such that $\lab(q)+C \preceq (t,u)$. Again, w.l.o.g. we can assume that $q$ is not a leaf as otherwise there exists an ancestor $q'$ of $q$ such that $\lab(q') \preceq \lab(q)$. By definition of $T$, for all $(t,t') \in E$, there is $q' \in Q$ such that $(q,q') \in R$ and $\lab(q') = \lab(q) + \weight(t, t')$. Let $(t', v) =  \lab(q') + C$. By definition of $f$, we have $(t',v) \in f$. By Eq.~\eqref{eq:cpre}, it follows that $(t,u) \in \Cpre_{\maxE}(f)$.

Now, let us show that $f \subseteq \Cpre_{\maxE}^{\ast}$. This is a direct consequence of the monotonicity of $\Cpre_{\maxE}$: it is well known that for any monotone function on a complete lattice, its greatest fixed point is equal to the least upper bound of all post-fixed points (points $e$ such that $e \subseteq \Cpre_{\maxE}(e)$), i.e., $\Cpre_{\maxE}^{\ast} = \bigcup \{ e \mid  e \subseteq \Cpre_{\maxE}(e) \} \supseteq f$. As $(\initState, (\maxE, \ldots{}, \maxE)) \in f$, that concludes the proof.\qed
\end{proof}

\begin{remark}
Note that the exponential bound on memory, obtained in Lemma \ref{lemma_width}, can also be derived from the Moore machine construction of Lemma~\ref{lem:MooreFP} as this method is complete according to Lemma~\ref{lem:gfpcpre}. 
Still, the DAG construction of Lemma \ref{lemma_width} is interesting in its own right, and introduces the concept of node merging, which is underlying to the symbolic algorithm correctness, while transparent in its use.
\end{remark}

\paragraph{{\bf Incrementality.}}
While the threshold $2 \cdot l \cdot W$ is sufficient, it may be the case that $\playerOne$ can win the game even if its energy level is bounded above by some smaller value. So, in practice, we can use Lemma \ref{lem:MooreFP}, to justify an incremental algorithm that first starts with small values for the parameter $\maxE$ and stops as soon as a winning strategy is found or when the value of $\maxE$ reaches the threshold $2\cdot l \cdot W$ and no winning strategy has been found.

\paragraph{{\bf Application of the symbolic algorithm to MEPGs and MMPGs.}}
Using the reduction of Lemma~\ref{lemma_parityToEnergy} that allows us to remove the parity condition, and the equivalence between multi energy games and multi mean-payoff games for finite-memory strategies (given by \cite[Theorem 3]{chatterjee_FSTTCS10}), along with Lemma \ref{lemma_cpre} (complexity), Lemma \ref{lem:MooreFP} (correctness) and Lemma \ref{lem:gfpcpre} (completeness), we obtain the following result.

\begin{theorem}[Symbolic and incremental synthesis algorithm]
\label{thm_symb}
Let $\gamePar$ be a multi energy (resp. multi mean-payoff) parity game. Algorithm $\algoName$ is a symbolic and incremental algorithm that synthesizes a winning strategy in $\gamePar$ of at most exponential size memory, if a winning (resp. finite-memory winning) strategy exists. In the worst-case, the algorithm $\algoName$ takes exponential time.
\end{theorem}

\begin{proof}
The correctness and completeness for algorithm $\algoName$ on multi energy games are resp. given by Lemma \ref{lem:MooreFP} and Lemma \ref{lem:gfpcpre}. Extension to mean-payoff games (under finite memory) is given by \cite[Theorem 3]{chatterjee_FSTTCS10}, whereas the parity condition can be encoded as energy thanks to Lemma \ref{lemma_parityToEnergy}. Exponential worst-case complexity of the algorithm $\algoName$ is induced by Lemma \ref{lemma_cpre}.\qed
\end{proof}

\paragraph{{\bf Integration in synthesis tools.}} Following the conference version of this paper \cite{chatterjee_CONCUR2012}, our results on strategy synthesis have been used in the {\sf Acacia+} synthesis tool. This tool originally handled the synthesis of controllers for specifications expressed in {\sf LTL} (Linear Temporal Logic, a classical formalism for formal specifications \cite{pnueli_FOCS1977}) using antichain-based algorithms and has recently been extended to the synthesis from {\sf LTL} specifications with mean-payoff objectives~\cite{bohy_TACAS2013}. The addition of multi mean-payoff objectives to {\sf LTL} specifications provides a convenient way to enforce that synthesized controllers also satisfy some reasonable behavior from a quantitative standpoint, such as minimizing the number of unsollicited grants in a client-server architecture with prioritized clients. Numerous practical applications may benefit from this multi-dimension framework.

The authors present an approach in which the corresponding synthesis problem ultimately reduces to strategy synthesis on a multi energy game \cite[Theorem 26]{bohy_TACAS2013}. Their implementation uses fixed point computations similar to Eq.~\eqref{eq:desc-chain} and has proved efficient (considering the complexity of the problem) in practice. It uses antichains to provide a compact representation of upper-closed sets and implements the incremental approach proposed before (regarding the constant $\maxE$). In practical benchmarks, winning strategies can generally be found for rather small values of $\maxE$. Hence, the incremental approach overcomes the need to compute up to the exponential theoretical bound $\maxE = 2 \cdot l \cdot W$ in many cases. Sample benchmarks and experiments can be found in \cite{bohy_TACAS2013}, and the tool can be used online \cite{acaciaWeb}.

\section{Trading finite memory for randomness}
\label{randomized}

In this section, we answer the fundamental question regarding the trade-off of memory for randomness in strategies: we study on which kind of games $\playerOne$ can replace a pure finite-memory winning strategy by an equally powerful, yet conceptually simpler, randomized memoryless one and discuss how memory is encoded into probability distributions. Note that we do not consider wider strategy classes (e.g., randomized finite-memory), nor do we allow randomization for $\playerTwo$ (which on most cases is dispensable anyway). Indeed, we aim at a better understanding of the underlying mechanics of memory and randomization, in order to provide alternative strategy representations of practical use; not exploration of more complex games with wider strategy classes (Lemma \ref{lemma_rfm} shows a glimpse of it).

\vspace{-3mm}
\begin{table}[bht]
  \centering   
\begin{tabular}{|c||c|c|c|c|}
\hline & ~Multi energy and energy parity~ & ~Multi MP (parity)~ & ~MP parity~ \\ 
\hline\hline ~one-player~ & $\times$ & $\surd$ & $\surd$ \\ 
\hline ~two-player~ & $\times$ & $\times$ & $\surd$ \\
\hline 
\end{tabular} 
\vspace{2mm}
      \caption{When pure finite memory for $\playerOne$ can be traded for randomized memorylessness.}
\label{tab:randomized}
  \end{table}

We present an overview of our results in Tab. \ref{tab:randomized} and summarize them in Theorem \ref{thm_rdm}. Note that we do not consider the opposite implication, i.e., does there always exist a way of encoding a randomized memoryless strategy into an equivalent finite-memory one. In general, this is not the case even for classes of games where we can trade memory for randomness, and it can easily be witnessed on the one-player multi mean-payoff game depicted on Fig. \ref{fig:satis_expect}. Indeed, expectation $(1, 1)$ is achievable with a simple uniform distribution while it is not achievable with a pure, arbitrary high memory strategy (even infinite).

\begin{figure}[tb]
\centering
\scalebox{1}{\begin{tikzpicture}[->,>=stealth',shorten >=1pt,auto,node
    distance=2.5cm,bend angle=45,scale=0.75,font=\scriptsize]
    \tikzstyle{p1}=[draw,circle,text centered,minimum size=8mm]
    \node[p1]  (0)  at (0, 0) {$s_{1}$};
    \node[p1]  (1) at (-2, 0) {$s_{2}$};
    \node[p1]  (2) at (2, 0)  {$s_{3}$};
    \coordinate[shift={(0mm,5mm)}] (init) at (0.north);
    \path
    (init) edge (0)
    (1) edge [loop left] node [below, xshift=-2mm, yshift=-2mm] {$(2, 0)$} (1)
    (2) edge [loop right] node [below, xshift=2mm, yshift=-2mm] {$(0, 2)$} (2)
    (0) edge node [below] {$(0, 0)$} (1)
    (0) edge node [below] {$(0, 0)$} (2);
      \end{tikzpicture}}
      \caption{Randomization can replace memory, but not the opposite.}
\label{fig:satis_expect}
\end{figure}
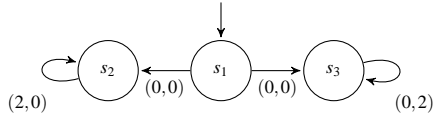
 
We break down these results into three subsections: energy games, multi mean-payoff (parity) games, and single mean-payoff parity games. We start with energy games.

\subsection{\textbf{Randomization and energy games}}
\label{app_energyRand}

Randomization is not helpful for energy objectives, even in one-player games. The proof argument is obtained from the intuition that energy objectives are similar in spirit to safety objectives.

\begin{lemma}
\label{lemma_enRand}
Randomization is not helpful for almost-sure winning in one-player and two-player energy, multi energy, energy parity and multi energy parity games: if there exists a finite-memory randomized winning strategy, then there exists a pure winning strategy with the same memory requirements.
\end{lemma}

\begin{proof}
Let $\gamePar$ be a game fitted with an energy objective. Consider an almost-sure winning strategy $\strat_{1}$. If there exists a single path $\play$ consistent with $\strat_{1}$ that violates the energy objective, then there exists a finite prefix witness $\prefix$ to violate the energy objective. Moreover, as the finite prefix has positive probability (otherwise the play is not consistent), and the strategy $\strat_{1}$ is almost-sure winning, it follows that no such path exists. In other words, $\strat_{1}$ is a sure winning strategy. Since randomization does not
help for sure winning strategy, it follows that randomization is not helpful for one-player and two-player energy, multi energy, energy parity and multi energy parity games.\qed
\end{proof}

\subsection{\textbf{Randomization and multi mean-payoff (parity) games}}
\label{app_mmppRand}

Randomized memoryless strategies can replace pure finite-memory ones in the one-player multi mean-payoff parity case, but not in the two-player one, even without parity. We first note a useful link between satisfaction and expectation semantics for the mean-payoff objective.

\begin{lemma}
\label{lemma_satis_expect}
Let $\gameFull$ be a game structure with mean-payoff objective $\objective = \objMPnoPar$ for some threshold vector $v \in \rat^{\dimension}$. Let $\strat_{1} \in \strats_{1}$ be a strategy of $\playerOne$. If $\strat_{1}$ is almost-sure winning for $\objective$ (i.e., winning for $1$-satisfaction), then $\strat_{1}$ is also winning for $v$-expectation for the mean-payoff function $\mpay$. The opposite does not hold.
\end{lemma}

\begin{proof}
We first discuss the claimed implication. Suppose $1$-satisfaction is verified. Then, for all strategy $\strat_{2} \in \strats_{2}$ of $\playerTwo$, the set of consistent plays of value $\geq v$ has measure $1$, while the one of value $< v$ has measure $0$, by definition. Therefore, the expectation $\expect_{s_{init}}^{\strat_{1}, \strat_{2}}(\mpay)$ is at least $v$ and $v$-expectation is verified.

To show that the opposite does not hold, consider the simple one-player game depicted on Fig. \ref{fig:satis_expect}. Let $\strat_{1}$ be a simple coin flipping on $s_{1}$, i.e., $\strat_{1}(s_{1})(s_{2}) = 1/2$, $\strat_{1}(s_{1})(s_{3}) = 1/2$, $\strat_{1}(s_{2})(s_{2}) = 1$ and $\strat_{1}(s_{3})(s_{3}) = 1$. The expectation of this strategy is $v = (1,1)$. Nevertheless, the probability of achieving mean-payoff of at least $v$ is $0 < 1$, which shows that it does not verify $1$-satisfaction for $\objMPnoPar$.\qed
\end{proof}
 
The fundamental difference between energy and mean-payoff is that energy requires a property to be satisfied \textit{at all times} (in that sense, it is similar to safety), while mean-payoff is a \textit{limit} property. As a consequence, what matters here is the long-run frequencies of weights, not their order of appearance, as opposed to the energy case.

\begin{lemma}
\label{lemma_multiMP}
Pure finite-memory winning strategies can be traded for equally powerful randomized memoryless ones for one-player multi mean-payoff parity games, for both satisfaction and expectation semantics. For two-player games, randomized memoryless strategies are not as powerful, even limited to expectation semantics, no parity condition, and only $2$ dimensions.
\end{lemma}

For the one-player case, we extract the frequencies of visit for edges of the graph from the regular outcome that arises from the finite-memory strategy of $\playerOne$. We build a randomized strategy with probability distributions on edges that yield the exact same frequencies in the long-run. Therefore, if the original pure finite-memory of $\playerOne$ is surely winning, the randomized one is almost-surely winning. For the two-player case, this approach cannot be used as frequencies are not well defined, since the strategy of $\playerTwo$ is unknown. Consider a game which needs perfect balance between frequencies of appearance of two sets of edges in a play to be winning (Fig. \ref{fig:multiMP}). To almost-surely achieve mean-payoff vector $(0, 0)$, $\playerOne$ must ensure that the long-term balance between edges $(s_{4}, s_{5})$ and $(s_{4}, s_{6})$ is the same as the one between edges $(s_{1}, s_{3})$ and $(s_{1}, s_{2})$. This is achievable with memory as it suffices to react immediately to compensate the choice of $\playerTwo$. However, given a randomized memoryless strategy of $\playerOne$, $\playerTwo$ always has a strategy to enforce that the long-term frequency is unbalanced, and thus the game cannot be won almost-surely by $\playerOne$ with such a strategy. Achieving expected mean-payoff $(0, 0)$ is also excluded.

\begin{figure}[htb]
\centering
  \scalebox{1}{\begin{tikzpicture}[->,>=stealth',shorten >=1pt,auto,node
    distance=2.5cm,bend angle=45,scale=0.6,font=\normalsize]
    \tikzstyle{p1}=[draw,circle,text centered,minimum size=8mm]
    \tikzstyle{p2}=[draw,rectangle,text centered,minimum size=7mm]
    \node[p2]  (1)  at (0, 0) {$s_{1}$};
    \node[p1]  (2) at (-2, -2) {$s_{2}$};
    \node[p1]  (3) at (2, -2)  {$s_{3}$};
    \node[p1]  (4) at (0, -4) {$s_{4}$};
    \node[p1]  (5) at (-5, -2)  {$s_{5}$};
    \node[p1]  (6) at (5, -2) {$s_{6}$};
    \coordinate[shift={(0mm,5mm)}] (init) at (1.north);
    \path
    (init) edge (1);
	\draw[->,>=latex] (1) to[out=225,in=45] node [left, xshift =-1mm, yshift =1mm] {$(1, -1)$} (2);
	\draw[->,>=latex] (1) to[out=315,in=135] node [right, xshift =1mm, yshift =1mm] {$(-1, 1)$} (3);
	\draw[->,>=latex] (2) to[out=315,in=135] node [left, xshift =-1mm, yshift =-1mm] {$(0, 0)$} (4);
	\draw[->,>=latex] (3) to[out=225,in=45] node [right, xshift =1mm, yshift =-1mm] {$(0, 0)$} (4);
	\draw[->,>=latex] (4) to[out=180,in=315] node [left, xshift =-5mm] {$(1, -1)$} (5);
	\draw[->,>=latex] (4) to[out=0,in=225] node [right, xshift =5mm] {$(-1, 1)$} (6);
	\draw[->,>=latex] (5) to[out=45,in=180] node [left, xshift =-5mm] {$(0, 0)$} (1);
	\draw[->,>=latex] (6) to[out=135,in=0] node [right, xshift =5mm] {$(0, 0)$} (1);
      \end{tikzpicture}}
      \caption{Memory is needed to enforce perfect long-term balance.}
      \label{fig:multiMP}
\end{figure}
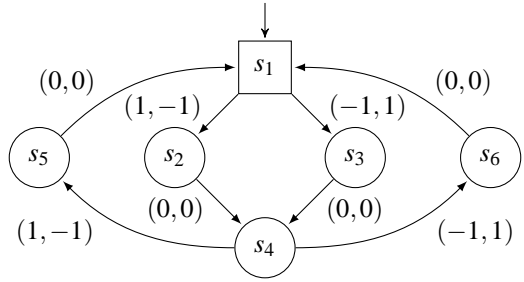

\begin{proof}
We begin with the one-player case. Let $\gamePar$ be a multi mean-payoff parity game. Let $\strat_{1}^{pf} \in \stratsPureFinite_{1}$ be the pure finite-memory strategy of the player. Since it is pure and finite, its outcome is a regular word $\play = \prefix_{1} \cdot (\prefix_{2})^{\omega}$, with $\prefix_{1} \in \states^{\ast}$, $\prefix_{2} \in \states^{+}$. Let $\objective = \objMP \cap \objPar$ be the multi mean-payoff parity objective for some threshold vector $v \in \rat^{\dimension}$. Suppose this strategy verifies $\alpha$-satisfaction for $\objective$ and $\beta$-expectation for the $\mpay$ function, for some $\alpha$, $\beta$. We claim that there exists a randomized memoryless strategy $\strat_{1}^{rm} \in \stratsRandomizedMemoryless_{1}$ that is also $\alpha$-satisfying for $\objective$ and that satisfies $\beta$-expectation for the $\mpay$ function; and we show how to build it.

We denote concatenation by the $\cdot$ symbol. Given a finite word $\prefix \in \states^{\ast}$, two states $s, s' \in \states$, we resp. denote by $\occ(s, \prefix)$ and $\occ((s, s'), \prefix)$ the number of occurences of the state $s$ and the transition $(s, s')$ in the word $\prefix$. We add the subscript $\circ$ when we count the first state of the word as the successor of the last one (i.e., the word is a cycle in the game graph). That is, $\occCirc(\ast, \prefix) = \occ(\ast, \prefix \cdot \first(\prefix))$.

Let us consider the mean-payoff of the outcome of strategy $\strat_{1}^{pf}$. Recall that for a play $\play \in \plays$, $\play = s^{1}, s^{2}, s^{3}\ldots{}$, we have $\mpay(\play) = \liminf_{n \rightarrow \infty} \frac{1}{n} \sum_{1 \leq i < n} \weight(s^{i}, s^{i+1})$.
Since the play induced by $\strat_{1}^{pf}$ is regular, the limit is well defined and we may express the mean-payoff in terms of frequencies, that is
\begin{equation*}
\mpay(\play) = \sum_{(s, s') \in \edges} \weight(s, s') \cdot \freq((s, s')),
\end{equation*}
where $\freq$ denotes the long-term frequency of a transition defined as
\begin{equation*}
\forall\, (s, s') \in \edges,\quad \freq((s, s')) = \dfrac{\occCirc((s, s'), \prefix_{2})}{\vert \prefix_{2}\vert}.
\end{equation*}

We define the randomized memoryless strategy $\strat_{1}^{rm}$ as follows: $\forall\, s, s' \in \states,\, (s, s') \in \edges,\, X = \left\lbrace (s, t) \,\vert\, t \in \states, (s, t) \in \left( \prefix_{1} \cdot \first(\prefix_{2})\right) \right\rbrace$,
\begin{align*}
\strat_{1}^{rm}(s)(s') = \begin{cases}\dfrac{1}{\vert X \vert} \text{ if } s \in \prefix_{1} \,\wedge\, s \not\in \prefix_{2},\\[1em]
\dfrac{\occCirc((s, s'), \prefix_{2})}{\occ(s, \prefix_{2})} \text{ if } s \in \prefix_{2},\\[1em]
0 \text{ otherwise}.
\end{cases}
\end{align*}
Intuitively, we fix a uniform distribution over transitions of the finite prefix $\prefix_{1}$ as we only need to ensure reaching the bottom strongly connected component (BSCC) defined by $\prefix_{2}$ with probability $1$, and the relative frequencies in $\prefix_{1}$ do not matter (because these weights and priorities are negligible in the long run). On the contrary, we use the exact frequencies for transitions of $\prefix_{2}$ as they prevail long-term wise. Note that $\strat_{1}^{rm}$ is a correctly defined randomized memoryless strategy.

Obviously, $\strat_{1}^{rm}$ yields a Markov chain over states of $(\prefix_{1} \cup \prefix_{2})$ such that states of $(\prefix_{1} \setminus \prefix_{2})$ are transient and states of $\prefix_{2}$ constitute a BSCC that is reached with probability one. Thus, the mean-payoff induced by $\strat_{1}^{rm}$ is totally dependent on this BSCC mean-payoff value. As a consequence, proving that transition frequencies in the BSCC are exactly the same as frequencies $\freq$ defined by $\strat_{1}^{pf}$ will imply the claim on mean-payoff. Moreover, parity will remain satisfied as the sets of infinitely often visited states will be the same for both the pure and the randomized strategy. Let $T = \lbrace t_{1}, t_{2}, \ldots{}, t_{m}\rbrace$ be the set of states that appear in $\prefix_{2}$. This BSCC is an ergodic Markov chain $\ergodic = (T, P)$ with the following matrix of transition probabilities:
\begin{equation*}
P = \begin{blockarray}{p{1cm}ccc}
  & t_{1} & \dots & t_{m}  \\
  \begin{block}{c(ccc)}
  t_{1} & \dfrac{\occCirc((t_{1}, t_{1}), \prefix_{2})}{\occ(t_{1}, \prefix_{2})} &  &  \\
  \vdots & & \ddots & &      \\
  t_{m} &   &  & \dfrac{\occCirc((t_{m}, t_{m}), \prefix_{2})}{\occ(t_{m}, \prefix_{2})} \\
  \end{block}
\end{blockarray}\quad.
\end{equation*} 
Classical analysis of ergodic Markov chains grants the existence of a unique probability vector $\nu$ such that $\nu P = \nu$, i.e.,
\begin{equation*}
\forall\, 1 \leq i \leq m,\; \nu_{i} = \sum_{1 \leq j \leq m} \dfrac{\occCirc\left( (t_{j}, t_{i}), \prefix_{2}\right) }{\occ\left( t_{j}, \prefix_{2}\right) } \cdot \nu_{j}.
\end{equation*}
This vector $\nu$ represents the occurence frequency of each state in an infinite run over the Markov chain. It is easy to see that the unique probability vector $\nu$ that satisfies $\nu P = \nu$ is
\begin{equation*}
\nu = \left( \dfrac{\occ(t_{1}, \prefix_{2})}{\vert \prefix_{2}\vert}, \quad\ldots{}\quad, \dfrac{\occ(t_{m}, \prefix_{2})}{\vert \prefix_{2}\vert}\right).
\end{equation*}
Moreover, given a transition of the Markov chain, its frequency is simply the product of the frequency of its starting state by the probability of the transition when the chain is in this state: for all $t, t' \in T$, we have $\freq^{\ergodic}((t, t')) = \nu(t) \cdot P(t, t')$. By definition of $\nu$ and $P$, that is 
\begin{equation*}
\freq^{\ergodic}((t, t')) = \dfrac{\occCirc((t, t'), \prefix_{2})}{\vert \prefix_{2}\vert} = \freq((t, t')),
\end{equation*}
thus proving that the randomized strategy $\strat_{1}^{rm}$ \textit{almost-surely} yields the same mean-payoff and parity as the pure finite-memory one~$\strat_{1}^{pf}$. The expected value threshold is also verified by Lemma~\ref{lemma_satis_expect}.

Now it remains to show that this does not carry over to two-player games. Indeed, we show that randomized memoryless strategies cannot replace pure finite-memory ones for the expectation semantics, even without parity. By Lemma \ref{lemma_satis_expect}, this implies that it cannot be verified for $1$-satisfaction semantics either. Consider the game depicted on Fig. \ref{fig:multiMP}. Player $\playerOne$ has a pure finite-memory strategy $\strat_{1}^{pf}$ that ensures $\mpay(\play) \geq (0, 0)$, against all strategy $\strat_{2}$ of $\playerTwo$. This strategy is simply to take the opposite choice of $\playerTwo$: $\strat_{1}^{pf}(\ast s_{2}s_{4}) = s_{6}$ and $\strat_{1}^{pf}(\ast s_{3}s_{4}) = s_{5}$. Now suppose $\playerOne$ uses a randomized memoryless strategy $\strat_{1}^{rm}$ such that $\strat_{1}^{rm}(s_{4})(s_{5}) = p$ and $\strat_{1}^{rm}(s_{4})(s_{6}) = 1-p$, for some $p \in \left[ 0, 1\right] $. We claim that whatever the value of $p$, there exists a counter-strategy $\strat_{2}$ for $\playerTwo$ such that $\expect_{s_{1}}^{\strat_{1}^{rm}, \strat_{2}}(\mpay) \not\geq (0, 0)$. Suppose $p \geq 1/2$ and let $\strat_{2}(s_{1}) = s_{2}$. Then, we have
\begin{equation*}
\expect_{s_{1}}^{\strat_{1}^{rm}, \strat_{2}}(\mpay) = \dfrac{(1, -1) + \left[ p \cdot (1, -1) + (1-p) \cdot (-1, 1)\right] }{4} = \dfrac{1}{2}(p, -p) \not\geq (0, 0).
\end{equation*}
Now suppose $p < 1/2$ and let $\strat_{2}(s_{1}) = s_{3}$. Then, we have
\begin{equation*}
\expect_{s_{1}}^{\strat_{1}^{rm}, \strat_{2}}(\mpay) = \dfrac{(-1, 1) + \left[ p \cdot (1, -1) + (1-p) \cdot (-1, 1)\right] }{4} = \dfrac{1}{2}(p - 1, 1 - p) \not\geq (0, 0).
\end{equation*}
This shows that memory is needed to achieve the $(0, 0)$-expectation objective and concludes our proof.\qed
\end{proof}

\subsection{\textbf{Randomization and single mean-payoff parity games}}

Randomized memoryless strategies can replace pure finite-memory ones for single mean-payoff parity games. The proof outline is as follows. We do it in two steps. First, we show that it is the case for the simpler case of \textit{MP Büchi games} (Lemma \ref{lemma_buchi}). Suppose $\playerOne$ has a pure finite-memory winning strategy for such a game. We use the existence of particular pure memoryless strategies on winning states: the classical attractor for B\"uchi states, and a strategy that ensures that cycles of the outcome have positive energy (whose existence follows from~\cite{chatterjee_ICALP10}). We build an almost-surely randomized memoryless winning strategy for $\playerOne$ by mixing those strategies in the probability distributions, with sufficient probability over the strategy that is good for energy. We illustrate this construction on the simple game $\gamePar$ depicted on Fig. \ref{fig:randBuchi}. Let $\strat_{1}^{pf} \in \stratsPureFinite_{1}$ be a strategy of $\playerOne$ such that $\playerOne$ plays $(s_{1}, s_{1})$ for $8$ times, then plays $(s_{1}, s_{2})$ once, and so on. This strategy ensures surely winning for the objective $\objective = \textsf{MeanPayoff}_{\gamePar}(3/5) \cap \textsf{Buchi}_{\gamePar}(\{\state_{2}\})$. Obviously, $\playerOne$ has a pure memoryless strategy that ensures winning for the B\"uchi objective: playing $(s_{1}, s_{2})$. On the other hand, he also has a pure memoryless strategy that ensures cycles of positive energy: playing $(s_{1}, s_{1})$. Let $\strat_{1}^{rm} \in \stratsRandomizedMemoryless_{1}$ be the strategy defined as follows: play $(s_{1}, s_{2})$ with probability $\gamma$ and $(s_{1}, s_{1})$ with the remaining probability. This strategy is almost-surely winning for $\objective$ for sufficiently small values of $\gamma$ (e.g., $\gamma = 1/9$). Second, we extend this result to \textit{MP parity games} using an induction on the number of priorities and the size of games (Lemma \ref{lemma_randPar}). We consider \textit{subgames} that reduce to the MP B\"uchi and MP coB\"uchi cases. For MP coB\"uchi games, pure memoryless strategies are known to suffice \cite{chatterjee_LICS05}.

\vspace{-3mm}
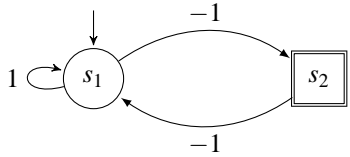
\begin{figure}[htb]
\centering
  \scalebox{1}{\begin{tikzpicture}[->,>=stealth',shorten >=1pt,auto,node
    distance=2.5cm,bend angle=45,scale=0.75,font=\normalsize]
    \tikzstyle{p1}=[draw,circle,text centered,minimum size=8mm]
    \tikzstyle{p2}=[draw,rectangle,text centered,minimum size=7mm,double]
    \node[p1]  (0)  at (0, 0) {$s_{1}$};
    \node[p2]  (1) at (4, 0) {$s_{2}$};
    \coordinate[shift={(0mm,5mm)}] (init) at (0.north);
    \path
    (0) edge [loop left] node [left] {$1$} (0)
    (init) edge (0);
	\draw[->,>=latex] (0) to[out=35,in=145] node [above] {$-1$} (1);
	\draw[->,>=latex] (1) to[out=215,in=325] node [below] {$-1$} (0);
      \end{tikzpicture}}
      \vspace{-1mm}
      \caption{Mixing strategies that are resp. \textit{good for B\"uchi} and \textit{good for energy}.}
\label{fig:randBuchi}
\end{figure}
\vspace{-6mm}

\paragraph{{\bf Büchi case.}} A particular, simpler case of the parity objective is the Büchi objective. It corresponds to parity with priorities $\lbrace 0, 1\rbrace$. We denote a Büchi game by $\game = (\states_{1}, \states_{2}, \initState, \edges, \weight, F)$, with $F$ the set of Büchi states such that a play is winning if it visits infinitely often states of the set $F$. We first state results on these Büchi objectives, as they are conceptually simpler to understand. Proof arguments for parity are more involved and make use of results on Büchi objectives. We sometimes denote the Büchi objective for the set $F$ by $\square\diamondsuit F$ (where $\square$ stands for \textit{globally} and $\diamondsuit$ for \textit{finally}), using the classical {\sf LTL} formulation \cite{pnueli_FOCS1977}.

We first introduce the useful notion of $\varepsilon$-optimality. Given a game $\gamePar$ with a one-dimension\footnote{The multi-dimensional setting gives rise to incomparable outcomes and the need to consider \textit{Pareto-optimality}.} mean-payoff objective, we define its value as
\begin{equation*}
\underline{\mpvalue} = \sup_{\strat_{1} \in \strats_{1}} \inf_{\strat_{2} \in \strats_{2}} \lbrace v \,\vert\, \outcomePar(\strat_{1}, \strat_{2}) \subseteq \objMP\rbrace.
\end{equation*}
A strategy is said optimal for the mean-payoff objective if it achieves this value. Such a strategy may not need to exist in general, even in one-player games \cite{chatterjee_LICS05,bouyer_ATVA11,chatterjee_MEMICS11} (Fig. \ref{fig:mpp}, $\playerOne$ has to delay its visits of $s_{1}$ for longer and longer intervals in order to tend towards value 1). However, it is known that for all $\varepsilon > 0$, $\varepsilon$-optimal strategies (i.e., that achieve value $(\underline{\mpvalue} - \varepsilon)$) always exist in one-dimension mean-payoff games, as a consequence of Martin's theorem on Borel determinacy \cite{martin_AM75}.

\vspace{-5mm}
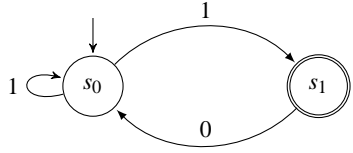
\begin{figure}[htb]
\centering
  \scalebox{1}{\begin{tikzpicture}[->,>=stealth',shorten >=1pt,auto,node
    distance=2.5cm,bend angle=45]
    \tikzstyle{p1}=[draw,circle,text centered,minimum size=8mm]
    \tikzstyle{p2}=[draw,circle,text centered,minimum size=8mm,double]
    \node[p1]  (0)  at (0, 0) {$s_{0}$};
    \node[p2]  (1) at (3, 0) {$s_{1}$};
    \coordinate[shift={(0mm,5mm)}] (init) at (0.north);
    \path
    (init) edge (0)
    (0) edge [loop left] node [left] {$1$} (0);
	\draw[->,>=latex] (0) to[out=45,in=135] node [above] {$1$} (1);
	\draw[->,>=latex] (1) to[out=225,in=315] node [above] {$0$} (0);
      \end{tikzpicture}}
      \vspace{-1mm}
      \caption{Mean-payoff Büchi requires infinite memory for optimality.}
\label{fig:mpp}
\end{figure}
\vspace{-3mm}

Here, we show finite-memory strategies can be traded off for randomized memoryless ones for mean-payoff Büchi games. Precisely, we prove that $\varepsilon$-optimality for mean-payoff Büchi games can as well be achieved by randomized memoryless strategies. We first need to state two useful lemmas granting the existence of pure memoryless strategies that are resp. \textit{good-for-energy} or \textit{good-for-Büchi}, in all states that are winning for the mean-payoff Büchi objective. These strategies will help us build the needed $\varepsilon$-optimal strategies.

\begin{lemma}[{Extension of \cite[Lemma 4]{chatterjee_ICALP10}}]
\label{lemma_gfe}
Let $\game = (\states_{1}, \states_{2}, \initState, \edges, \weight, F)$, with $F$ the set of Büchi states. Let $\winEP \subseteq \states$ be the set of winning states for the mean-payoff Büchi objective with threshold $0$. For all $s \in \winEP$, $\playerOne$ has a uniform (i.e., independent of the starting state) memoryless good-for-energy strategy $\strat_{1}^{\textit{gfe}}$ whose outcome never leaves the set $\winEP$, such that any cycle $c$ of this outcome has energy $\el(c) \geq 0$.
\end{lemma}

\begin{lemma}[Classical attractor]
\label{lemma_gfb}
Let $\game = (\states_{1}, \states_{2}, \initState, \edges, \weight, F)$, with $F$ the set of Büchi states. Let $\winEP \subseteq \states$ be the set of winning states for the mean-payoff Büchi objective with threshold $0$. For all $s \in \winEP$, $\playerOne$ has a uniform (i.e., independent of the starting state) memoryless good-for-Büchi strategy $\strat_{1}^{\diamondsuit F}$, an attractor strategy for $F$, whose outcome never leaves the set $\winEP$, such that  it ensures reaching $F$ in at most $\vert \states\vert$ steps.
\end{lemma}

The randomized memoryless strategy of $\playerOne$ will thus consist in mixing these two strategies, with a very low probability on the good-for-Büchi strategy. Indeed, the Büchi objective will be satisfied whatever this probability is, provided it is strictly positive. On the other hand, by giving more weight to the good-for-energy strategy, $\playerOne$ can obtain a mean-payoff that is arbitrary close to the optimum.

\begin{lemma}
\label{lemma_buchi}
In mean-payoff Büchi games, $\varepsilon$-optimality can be achieved surely by pure finite-memory strategies and almost-surely by randomized memoryless strategies.
\end{lemma}

\begin{proof}
Let $\game = (\states_{1}, \states_{2}, s_{init}, \edges, \weight, F)$, with $F$ the set of Büchi states. We consider the mean-payoff objective with threshold $0$ (w.l.o.g.). Let $\winEP \subseteq \states$ be the set of winning states for the mean-payoff Büchi objective. By Lemma \ref{lemma_gfe} and Lemma \ref{lemma_gfb}, for all $s \in \winEP$, $\playerOne$ has two uniform memoryless strategies $\strat_{1}^{\textit{gfe}}$ and $\strat_{1}^{\diamondsuit F}$, whose outcomes never leave the set $\winEP$, such that $\strat_{1}^{\textit{gfe}}$ ensures that any cycle $c$ of its outcome has energy $\el(c) \geq 0$, and $\strat_{1}^{\diamondsuit F}$, an attractor strategy for $F$, ensures reaching $F$ in at most $\vert \states\vert$ steps.

We first build $\varepsilon$-optimal \textit{pure finite-memory} strategies based on these two pure memoryless strategies. Let $\varepsilon > 0$. As usual, $\largestW$ denotes the largest absolute weight on any edge. Let us define $\strat_{1}^{pf}$ such that (a) it plays $\strat_{1}^{\textit{gfe}}$ for $\left\lceil\frac{2 \cdot \largestW \cdot \vert \states\vert}{\varepsilon}\right\rceil - \vert \states\vert$ steps, then (b) it plays $\strat_{1}^{\diamondsuit F}$ for $\vert \states\vert$ steps, then again (a). This ensures that $F$ is visited infinitely often as $\strat_{1}^{\diamondsuit F}$ is played infinitely many times for $\vert \states\vert$ steps in a row. Furthermore, the total cost of phases (a) + (b) is bounded by $-2\cdot \largestW \cdot \vert \states\vert$, and thus the mean-payoff of the outcome is at least $-\varepsilon$, against any strategy of the adversary.

Second, we show that based on the same pure memoryless strategies, it is possible to obtain almost-surely $\varepsilon$-optimal \textit{randomized memoryless} strategies, i.e.,
\begin{align*}
&\forall\, \varepsilon > 0,\; \exists\, \strat_{1}^{rm} \in \stratsRandomizedMemoryless_{1},\; \forall\, \strat_{2} \in \strats_{2},\\
&\proba^{\strat_{1}^{rm}, \strat_{2}}_{\initState} \left(\play \vDash \square\diamondsuit F\right) = 1 \;\wedge\; \proba^{\strat_{1}^{rm}, \strat_{2}}_{\initState} \left( \mpay(\play) \geq - \varepsilon\right) = 1.
\end{align*}
Note that pure memoryless strategies suffice for $\playerTwo$ as he essentially has to win against the Büchi \textit{or} the mean-payoff criterion \cite{bouyer_ATVA11}. Therefore, given $\varepsilon > 0$, we need to build some strategy $\strat_{1}^{rm} \in \stratsRandomizedMemoryless_{1}$ such that
\begin{equation*}
\forall\, \strat_{2}^{pm} \in \stratsPureMemoryless_{2},\;\proba^{\strat_{1}^{rm}, \strat_{2}^{pm}}_{\initState} \left(\play \vDash \square\diamondsuit F\right) = 1 \;\wedge\; \proba^{\strat_{1}^{rm}, \strat_{2}^{pm}}_{\initState} \left( \mpay(\play) \geq - \varepsilon\right) = 1.
\end{equation*}
We build such a strategy as follows:
\begin{equation*}
\forall s \in \states,\, \strat_{1}^{rm}(s) = \begin{cases}\strat_{1}^{\textit{gfe}}(s) \text{ with probability } 1 - \pr,\\ \strat_{1}^{\diamondsuit F}(s) \text{ with probability } \pr,\end{cases}
\end{equation*}
for some \textit{well-chosen} $\pr \in \left] 0, 1\right[ $.

It is straightforward to see that the Büchi objective is almost-surely satisfied for all values of $\pr > 0$ as at all times, the probability of playing according to $\strat_{1}^{\diamondsuit F}$ for $\vert \states \vert$ steps in a row, and thus ensuring a visit of $F$, is $\pr^{\vert \states\vert}$, which is strictly positive.

It remains to study if it is always possible to choose such a constant $\pr$ such that objective $\textsf{MeanPayoff}_{\gamePar}(-\varepsilon)$ is almost-surely satisfied. Consider such a strategy $\strat_{1}^{rm} \in \stratsRandomizedMemoryless_{1}$ and some fixed strategy $\strat_{2}^{pm} \in \stratsPureMemoryless_{2}$ of $\playerTwo$: the game reduces to the finite Markov chain $\chain = \left( \states, \delta, \weight\right) $, where $\delta \colon \edges \rightarrow \left[ 0, 1\right]$ is the transition probability function resulting from fixing those strategies. Suppose $\strat_{2}^{pm}$ is winning for $\playerTwo$. Thus, $\proba^{\chain}_{\initState}  \left( \mpay(\play) < - \varepsilon\right) > 0$. The mean-payoff depends on limit behavior: the probability measure of plays that do not enter in a bottom strongly connected component (BSCC) is zero~\cite{baier_MIT08}, whereas in a BSCC, the expected mean-payoff is the same in all states and it is obtained almost-surely (as follows from definition of BSCCs and prefix-independence of the mean-payoff). This implies that there exists some BSCC $\ec$ in $\chain$ such that $\proba^{\chain}_{\initState} \left(\diamondsuit \ec\right) > 0$ and $\expect^{\ec} \left(\mpay\right) < -\varepsilon$.

We claim that it is possible to choose $\pr$ such that all BSCCs, in all Markov chains induced by pure memoryless strategies of $\playerTwo$, have expectation greater than or equal to $\varepsilon$, thus proving that strategy $\strat_{1}^{rm}$ is almost-surely $\varepsilon$-optimal with regard to the mean-payoff value function. Intuitively, the smaller this constant $\pr$ is chosen, the nearer will the expected mean-payoff induced by $\strat_{1}^{rm}$ be to the one induced by $\strat_{1}^{\textit{gfe}}$, that is at least zero. Since the number of pure memoryless strategies of~$\playerTwo$ is finite, and so is the number of BSCCs induced by $\strat_{1}^{rm}$ (regardless of the exact value of $\pr \in \left] 0, 1\right[$, we obtain the same BSCCs in terms of states and edges), one can compute a suitable $\pr$ for each of them, and take the mininum to ensure that the property will be satisfied in all possible cases.

Therefore, let us fix some strategy $\strat_{2}^{pm}$ of $\playerTwo$, and some BSCC $\ec$  of the induced Markov chain when played against strategy $\strat_{1}^{rm}$ of $\playerOne$. It remains to show that (\textit{claim}) there exists $\pr \in \left] 0, 1\right]$ such that $\expect^{\ec(\pr)}(\mpay) \geq -\varepsilon$ to conclude this proof. Observe that we write $\ec(\pr)$ as transition probabilities inside $\ec$ depend on $\pr$.
By contradiction, suppose the claim is false. Precisely, we assume that (\textit{contradiction hypothesis}) for all $\pr \in \left] 0, 1\right]$, we have that $\expect^{\ec(\pr)}(\mpay) < -\varepsilon$.

Besides, observe that for $\pr = 0$, strategy $\strat_{1}^{rm}$ is exactly equal to $\strat_{1}^{\textit{gfe}}$. As we know that $\strat_{1}^{\textit{gfe}}$ ensures a worst-case mean-payoff at least equal to zero, we trivially deduce that $\expect^{\ec(0)} \geq 0$. This implies that $\sup_{\pr \in [0, 1]} \expect^{\ec(\pr)} \geq \expect^{\ec(0)} \geq 0$.
Notice that in this case, interval $[0, 1]$ is closed.

By results in the literature, it is known that this supremum is continuous. See for example Solan~\cite{solan} on the continuity of the optimal expected value function in the general context of competitive Markov decision processes (equivalent to $2\frac{1}{2}$-player games). Therefore, we have that
$\sup_{\pr \in ]0, 1]} \expect^{\ec(\pr)} = \sup_{\pr \in [0, 1]} \expect^{\ec(\pr)} \geq \expect^{\ec(0)} \geq 0$.
On the other hand, by (\textit{contradiction hypothesis}), we also have that
$\sup_{\pr \in ]0, 1]} \expect^{\ec(\pr)} \leq -\varepsilon$.
Since $\varepsilon$ is strictly positive, there is a clear contradiction, which concludes our proof.\qed
\end{proof}

\paragraph{{\bf Parity Case.}} Given those results for mean-payoff Büchi games, we now consider the more general case of mean-payoff parity games. We start by introducing the useful notion of \textit{subgames}.

\smallskip\noindent\textit{Subgame.} Let $\gameParFull$ be a game and $A \subseteq \states$ be a subset of states in $\gamePar$. If $\edges$ is such that for all $s \in A$, there exists $s' \in A$ with $(s,s') \in \edges$, then we define the {\em subgame} $\gamePar \downarrow A$ as $(\states_{1} \cap A, \states_{2} \cap A, \edges \cap (A \times A), \weight', p')$ where $\weight'$, $p'$ are the functions $\weight$, $p$ restricted to the subdomain $A$. Note that for subgames, we do not consider an initial state.

Let $\gameParFull$ and $U \subseteq \states$. We define $\Attr_1(U)$ as the set that is obtained as the limit of the following increasing sequence: $U_0=U$, and $U_i=U_{i-1} \cup \{ s \in \states_{1} \mid \exists\, s' \in U_{i-1},\, (s, s') \in \edges \} \cup \{ s \in \states_{2} \mid \forall\, s',\, (s, s') \in \edges,\, s' \in U_{i-1} \}$, for $i \geq 1$. As this sequence of sets is increasing, there exists $i \leq \vert\states\vert$ such that $U_j=U_i$ for all $j \geq i$. $\Attr_1(U)$ contains all the states in $\game$ from which $\playerOne$ can force a visit to $U$, and it is well known that $\playerOne$ has a pure memoryless strategy to force such a visit from those states. Also, it is clear that $\playerOne$ does not have a strategy to leave the states in $\states \setminus \Attr_1(U)$. Attractors can be defined symmetrically for $\playerTwo$ and are noted $\Attr_2(\cdot)$. As direct consequence, we have the following proposition.

\begin{proposition}
\label{attract-subgame}
Let $\gameParFull$ be a game, let $U \subseteq \states$ and $\Attr_1(U)$ be such that $B = \states \setminus \Attr_1(U)$ is non-empty, then $\gamePar \downarrow B$ is a subgame.
\end{proposition}

The following lemma states that optimal pure memoryless strategies exist for $\playerOne$ in games with mean-payoff coB\"uchi objectives (i.e., parity with priorities $\lbrace 1, 2\rbrace$). For mean-payoff Büchi objectives, we showed in Lemma \ref{lemma_buchi} that, for all $\varepsilon > 0$, $\varepsilon$-optimal randomized memoryless strategies exist.

\begin{lemma}[{\cite[Theorem 5]{chatterjee_LICS05}}]
\label{lemma_cobuchi}
Let $\gameParFull$ be a game with priorities $\lbrace 1, 2\rbrace$, and $\winningNodes^p_{\geq 0}$ be the set of nodes in $\gamePar$ from which $\playerOne$ wins the mean-payoff coB\"uchi objective for threshold $0$ (w.l.o.g.). Then from all states in $\winningNodes^p_{\geq 0}$, $\playerOne$ has a pure memoryless winning strategy for the coB\"uchi mean-payoff objective for threshold $0$.
\end{lemma}

We now establish that $\varepsilon$-optimal randomized memoryless strategies also exist for mean-payoff \textit{parity} games, and thus, can replace pure finite-memory ones.

\begin{lemma}
\label{lemma_randPar}
Let $\gameParFull$ and $\winningNodes^p_{\geq 0}$ be the set of nodes in $\gamePar$ from which $\playerOne$ wins the mean-payoff parity objective for threshold $0$. Then for all $\varepsilon >0$, there exists $\strat_1^{rm} \in \stratsRandomizedMemoryless_{1}$, such that for all $s \in \winningNodes^p_{\geq 0}$ and for all $\strat_2 \in \strats_2$, we have that:
\begin{equation*}
\proba^{\strat_{1}^{rm}, \strat_{2}}_{s} \left( \mpay(\play) \geq - \varepsilon\right) = 1 \;\wedge\; \proba^{\strat_{1}^{rm}, \strat_{2}}_{s} \left( \minpar(\play) \text{ mod } 2 = 0 \right) = 1.
\end{equation*}
\end{lemma}

\begin{proof}
The proof is by induction on the lexicographic order $\nodeLess$ on games, defined as follows: $\gamePar^1 \nodeLess \gamePar^2$ if $\gamePar^1$ has less priorities than $\gamePar^2$ or $\gamePar^1$ has the same priorities than in $\gamePar^2$ but less states. Clearly, this lexicographic order is well-founded. 

The base cases are twofold: one for the number of states, and one for priorities. First, if the game is such that $\vert\states\vert = 1$, then obviously, if $\playerOne$ can win, he can do so with a pure memoryless strategy, which respects the claim. Second, for two priorities. W.l.o.g., we can assume that all priorities are either in $\{0,1\}$ or in $\{1,2\}$. Those cases resp. correspond to mean-payoff B\"uchi and mean-payoff coB\"uchi games. The result for mean-payoff B\"uchi games has been established in Lemma \ref{lemma_buchi}, while the result for mean-payoff coB\"uchi games is a direct consequence of Lemma \ref{lemma_cobuchi}, as pure memoryless strategies are a special case of randomized memoryless strategies.

Let us now consider the inductive case. Suppose we have a mean-payoff parity game $\gamePar$ with $\priorities$ priorities and $\vert\states\vert$ states. W.l.o.g., we can make the assumption that the lowest priority in $\gamePar$ is either $0$ or $1$, otherwise we subtract an even number to all priorities so that we are in that case. Let $U_0=\{ s \in \winningNodes^p_{\geq 0} \mid p(s)=0 \}$ and $U_1=\{ s \in \winningNodes^p_{\geq 0} \mid p(s)=1 \}$.

\vspace{4mm}\noindent We consider the two possible following situations corresponding to $U_0$ empty or not.
\begin{enumerate}
	\item \textit{$U_0$ empty.} In that case $U_1$ is not empty. Let us consider $A_2=\Attr_2(U_1)$ the attractor of $\playerTwo$ for $U_1$. It must be the case that $\winningNodes^p_{\geq 0} \setminus A_2$ is non-empty, otherwise this would contradict the fact that $\playerOne$ is winning the parity objective from states in $\winningNodes^p_{\geq 0}$. Indeed, if it was not the case, then $\playerTwo$ would be able to force an infinite number of visits to $U_1$ from all states in $\winningNodes^p_{\geq 0}$, and the parity would be odd as $U_{0}$ is empty, a contradiction with the definition of $\winningNodes^p_{\geq 0}$. 
\textit{(i)} Let $B= \winningNodes^p_{\geq 0} \setminus A_2$. First note that, as $B$ is non-empty, by Proposition \ref{attract-subgame}, $\gamePar \downarrow B$ is a subgame. Also, note that from all states in $B$, it must be the case that $\playerOne$ has a winning strategy that does not require visits of the states outside $B$, i.e., states in $A_2$, for otherwise this would lead to a contradiction with the fact that $\playerOne$ is winning the parity objective in $\winningNodes^p_{\geq 0}$. So all states in the subgame $\gamePar \downarrow B$ are winning for $\playerOne$. The game $\gamePar \downarrow B$ does not contain states with priority $0$, and so we can apply our induction hypothesis to conclude that $\playerOne$ has a memoryless randomized strategy from all states in $B$, as $(\gamePar \downarrow B) \nodeLess \gamePar$ since it has one less priority. \textit{(ii)} Now, let us concentrate on states in $A_2$. Let $A_1=\Attr_1(B)$. From states in $A_1$, $\playerOne$ has a pure memoryless strategy to reach states in $B$, and so from there $\playerOne$ can play as in $\gamePar \downarrow B$, and we are done. Let $C=A_2 \setminus A_1$. If $C$ is empty, we are done. Otherwise, by Proposition \ref{attract-subgame}, $\gamePar \downarrow C$ is a subgame ($\playerTwo$ can force to stay within $C$). We conclude that all states in this game must be winning for $\playerOne$. This game has the same minimal priority than in the original game (i.e., priority $1$) but it has at least one state less, and so we can apply our induction hypothesis to conclude that $\playerOne$ has a memoryless randomized strategy from all states in $C$. Therefore, by \textit{(i)} and \textit{(ii)}, $\playerOne$ has a memoryless randomized strategy from all states in $\winningNodes^p_{\geq 0}$, which proves the claim in that case. 

\item \textit{$U_0$ is not empty.} Let us consider $A_1=\Attr_1(U_0)$. \textit{(iii)} First, consider the case where $A_1=\winningNodes^p_{\geq 0}$. In this case, it means that $\playerOne$ can force a visit to states in $U_0$ from any states in $\winningNodes^p_{\geq 0}$. So, we conclude that $\playerOne$ wins in $\gamePar$ the mean-payoff B\"uchi game with threshold $0$, and by Lemma \ref{lemma_buchi}, we conclude that $\playerOne$ has a memoryless randomized strategy from all states in $\gamePar$ for almost surely winning the parity game with mean-payoff threshold $0$ so we are done. \textit{(iv)} Second, consider the case where $B=\winningNodes^p_{\geq 0} \setminus A_1$ is non-empty. Then by Proposition~\ref{attract-subgame}, $\gamePar \downarrow B$ is a subgame. So $\playerTwo$ can force to stay within $B$ in the original game and so we conclude that all states in the game $\gamePar \downarrow B$ are winning for $\playerOne$. As $\gamePar \downarrow B$ does not contain states of priority $0$, and thus has at least one less priority, we can apply the induction hypothesis to conclude that $\playerOne$ has a memoryless randomized strategy from all states in $B$. Therefore, by \textit{(iii)} and \textit{(iv)}, $\playerOne$ has a memoryless randomized strategy from all states in $\winningNodes^p_{\geq 0}$, which also proves the case.
\end{enumerate}
As we have proved the claim in both possible cases, this concludes the proof.\qed
\end{proof}

\subsection{\textbf{Summary for randomization}}

We sum up results for these different classes of games in Theorem \ref{thm_rdm} (cf. Table \ref{tab:randomized}).

\begin{theorem}[Trading finite memory for randomness]
\label{thm_rdm}
The following assertions hold: (1) Randomized strategies are exactly as powerful as pure strategies for energy objectives. Randomized memoryless strategies are not as powerful as pure finite-memory strategies for almost-sure winning in one-player and two-player energy, multi energy, energy parity and multi energy parity games. (2) Randomized memoryless strategies are not as powerful as pure finite-memory strategies for almost-sure winning in two-player multi mean-payoff games. (3) In one-player multi mean-payoff parity games, and two-player single mean-payoff parity games, if there exists a pure finite-memory sure winning
strategy, then there exists a randomized memoryless almost-sure winning strategy.
\end{theorem}

\begin{proof}
(1) For energy games, results follow from Lemma \ref{lemma_enRand}. (2) For two-player multi mean-payoff games, they follow from Lemma \ref{lemma_multiMP}. (3) For one-player multi mean-payoff games, they follow from Lemma \ref{lemma_multiMP}. For two-player single mean-payoff parity, they are direct consequence of Lemma \ref{lemma_randPar}.\qed
\end{proof}

We close this section by observing that there are even more powerful classes of strategies. Their study, as well as their practical interest, remains open.

\begin{lemma}
\label{lemma_rfm}
Randomized finite-memory strategies are strictly more powerful than both randomized memoryless and pure finite-memory strategies for multi-mean payoff games with expectation semantics, even in the one-player case.
\end{lemma}

The intuition is essentially that memory permits to achieve an exact payoff by sticking to a given side, while randomization permits to combine payoffs of pure strategies to achieve any linear combination in between.

\begin{figure}[htb]
\centering
\scalebox{1}{\begin{tikzpicture}[->,>=stealth',shorten >=1pt,auto,node
    distance=2.5cm,bend angle=45]
    \tikzstyle{p1}=[draw,circle,text centered,minimum size=8mm]
    \node[p1]  (0)  at (0, 0) {$s_{0}$};
    \node[p1]  (1) at (3, 0) {$s_{1}$};
    \coordinate[shift={(-5mm,0mm)}] (init) at (0.west);
    \path
    (init) edge (0)
    (0) edge [loop above] node [above] {$(1, -1)$} (0)
    (1) edge [loop above] node [above] {$(-1, 1)$} (1)
    (0) edge node [above] {$(0, 0)$} (1);
      \end{tikzpicture}}
      \caption{Randomized finite memory is strictly more powerful than randomized memorylessness and pure finite memory.}
\label{fig:rfm}
\end{figure}
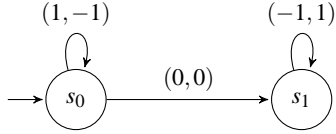 

\begin{proof}
Consider the game $\game$ depicted on Fig. \ref{fig:rfm}. Whatever the pure finite-memory strategy of $\playerOne$, the only achievable mean-payoff values are $(1,-1)$ (if $(s_{0}, s_{1})$ is never taken) and $(-1, 1)$ (if $(s_{0}, s_{1})$ is taken). This is also true for randomized memoryless strategies: either the probability of $(s_{0}, s_{1})$ is null and the mean-payoff has value $(1,-1)$, or this probability is strictly positive, and the mean-payoff has value $(-1,1)$ as the probability mass will eventually reach~$s_{1}$. On the contrary, value $(0, 0)$ is achievable by a randomized finite-memory strategy. Indeed, consider the strategy that tosses a coin in its first visit of $s_{0}$ to decide if it will always play $(s_{0}, s_{0})$ or if it will play $(s_{0}, s_{1})$ and then always $(s_{1}, s_{1})$. This strategy only needs one bit of memory and one bit to encode probabilities, and still, it is strictly more powerful than any amount of pure memory or any arbitrary high precision for probabilities without memory.\qed
\end{proof}

\section{Conclusion}
\label{conclusion}
In this work, we considered the finite-memory strategy synthesis problem for 
games with multiple quantitative (energy and mean-payoff) objectives along 
with a parity objective. 
We established tight (matching upper and lower) exponential bounds on the 
memory requirements for such strategies (Theorem~\ref{thm_optBounds}), 
significantly improving the previous triple exponential bound for multi energy games (without parity) that could be derived from results in literature for games on VASS. 
We presented an optimal symbolic and incremental strategy synthesis algorithm
(Theorem \ref{thm_symb}). As discussed in Section \ref{sec:alg}, the presented algorithm has been used as part of the synthesis tool {\sf Acacia+} for specifications combining {\sf LTL} properties and multi-dimensional quantitative objectives \cite{bohy_TACAS2013} and has proved efficient in practice. 
Finally, we also presented a precise characterization of the trade-off of 
memory for randomness in strategies (Theorem \ref{thm_rdm}).

\bibliographystyle{plain}
\bibliography{arXiv_bib}

\begin{thebibliography}{10}

\bibitem{acaciaWeb}
{\sf Acacia+}.
\newblock \url{http://lit2.ulb.ac.be/acaciaplus/}.

\bibitem{ACHMV10}
P.A. Abdulla, Y.-F. Chen, L.~Hol\'{\i}k, R.~Mayr, and T.~Vojnar.
\newblock When simulation meets antichains.
\newblock In {\em Proc. of TACAS}, LNCS 6015. Springer, 2010.

\bibitem{AHK02}
R.~Alur, T.A. Henzinger, and O.~Kupferman.
\newblock Alternating-time temporal logic.
\newblock {\em J. ACM}, 49(5):672--713, 2002.

\bibitem{baier_MIT08}
C.~Baier and J.-P. Katoen.
\newblock {\em Principles of model checking}.
\newblock MIT Press, 2008.

\bibitem{BernetJW02}
J.~Bernet, D.~Janin, and I.~Walukiewicz.
\newblock Permissive strategies: from parity games to safety games.
\newblock {\em ITA}, 36(3):261--275, 2002.

\bibitem{bloem_CAV09}
R.~Bloem, K.~Chatterjee, T.A. Henzinger, and B.~Jobstmann.
\newblock Better quality in synthesis through quantitative objectives.
\newblock In {\em Proc. of CAV}, LNCS 5643, pages 140--156. Springer, 2009.

\bibitem{BGHJ09}
R.~Bloem, K.~Greimel, T.A. Henzinger, and B.~Jobstmann.
\newblock Synthesizing robust systems.
\newblock In {\em Proc. of FMCAD}, pages 85--92. IEEE, 2009.

\bibitem{bohy_TACAS2013}
A.~Bohy, V.~Bruy\`ere, E.~Filiot, and J.-F. Raskin.
\newblock Synthesis from {LTL} specifications with mean-payoff objectives.
\newblock In {\em Proc. of TACAS}, LNCS 7795, pages 169--184. Springer, 2013.

\bibitem{borosh_AMS76}
I.~Borosh and B.~Treybig.
\newblock Bounds on positive integral solutions of linear diophantine
  equations.
\newblock {\em Proc. of the American Mathematical Society}, 55(2):299--304,
  1976.

\bibitem{bouyer_FORMATS2008}
P.~Bouyer, U.~Fahrenberg, K.G. Larsen, N.~Markey, and J.~Srba.
\newblock Infinite runs in weighted timed automata with energy constraints.
\newblock In {\em Proc. of FORMATS}, LNCS 5215, pages 33--47. Springer, 2008.

\bibitem{bouyer_ATVA11}
P.~Bouyer, N.~Markey, J.~Olschewski, and M.~Ummels.
\newblock Measuring permissiveness in parity games: Mean-payoff parity games
  revisited.
\newblock In {\em Proc. of ATVA}, LNCS 6996, pages 135--149. Springer, 2011.

\bibitem{brazdil_ICALP10}
T.~Br{\'a}zdil, P.~Jancar, and A.~Kucera.
\newblock Reachability games on extended vector addition systems with states.
\newblock In {\em Proc. of ICALP}, LNCS 6199, pages 478--489. Springer, 2010.

\bibitem{DBLP:conf/stacs/BruyereFRR14}
V.~Bruy{\`{e}}re, E.~Filiot, M.~Randour, and J.-F. Raskin.
\newblock Meet your expectations with guarantees: Beyond worst-case synthesis
  in quantitative games.
\newblock In {\em Proc. of STACS}, LIPIcs 25, pages 199--213. Schloss Dagstuhl
  - LZI, 2014.

\bibitem{CCHRS11}
P.~Cern{\'y}, K.~Chatterjee, T.A. Henzinger, A.~Radhakrishna, and R.~Singh.
\newblock Quantitative synthesis for concurrent programs.
\newblock In {\em Proc. of CAV}, LNCS 6806, pages 243--259. Springer, 2011.

\bibitem{CHR12}
P.~Cern{\'y}, T.A. Henzinger, and A.~Radhakrishna.
\newblock Simulation distances.
\newblock {\em Theor. Comput. Sci.}, 413(1):21--35, 2012.

\bibitem{CdAHS03}
A.~Chakrabarti, L.~de~Alfaro, T.A. Henzinger, and M.~Stoelinga.
\newblock Resource interfaces.
\newblock In {\em Proc. of EMSOFT}, LNCS 2855, pages 117--133. Springer, 2003.

\bibitem{chatterjee_ICALP10}
K.~Chatterjee and L.~Doyen.
\newblock Energy parity games.
\newblock In {\em Proc. of ICALP}, LNCS 6199, pages 599--610. Springer, 2010.

\bibitem{chatterjee_MEMICS11}
K.~Chatterjee and L.~Doyen.
\newblock Games and markov decision processes with mean-payoff parity and
  energy parity objectives.
\newblock In {\em Proc. of MEMICS}, LNCS. Springer, 2011.

\bibitem{CDH10}
K.~Chatterjee, L.~Doyen, and T.A. Henzinger.
\newblock Quantitative languages.
\newblock {\em ACM Trans. Comput. Log.}, 11(4), 2010.

\bibitem{chatterjee_FSTTCS10}
K.~Chatterjee, L.~Doyen, T.A. Henzinger, and J.-F. Raskin.
\newblock Generalized mean-payoff and energy games.
\newblock In {\em Proc. of FSTTCS}, LIPIcs 8, pages 505--516. Schloss Dagstuhl
  - LZI, 2010.

\bibitem{chatterjee_ATVA2013}
K.~Chatterjee, L.~Doyen, M.~Randour, and J.-F. Raskin.
\newblock Looking at mean-payoff and total-payoff through windows.
\newblock In {\em Proc. of ATVA}, LNCS 8172, pages 118--132. Springer, 2013.

\bibitem{chatterjee_LICS05}
K.~Chatterjee, T.A. Henzinger, and M.~Jurdzinski.
\newblock Mean-payoff parity games.
\newblock In {\em Proc. of LICS}, pages 178--187. IEEE Computer Society, 2005.

\bibitem{chatterjee_CONCUR2012}
K.~Chatterjee, M.~Randour, and J.-F. Raskin.
\newblock Strategy synthesis for multi-dimensional quantitative objectives.
\newblock In {\em Proc. of CONCUR}, LNCS 7454, pages 115--131. Springer, 2012.

\bibitem{DBLP:journals/acta/ChatterjeeRR14}
K.~Chatterjee, M.~Randour, and J.-F. Raskin.
\newblock Strategy synthesis for multi-dimensional quantitative objectives.
\newblock {\em Acta Informatica}, 51(3-4):129--163, 2014.

\bibitem{Church62}
A.~Church.
\newblock Logic, arithmetic, and automata.
\newblock In {\em Proceedings of the International Congress of Mathematicians},
  pages 23--35. Institut Mittag-Leffler, 1962.

\bibitem{InterfaceTheories}
L.~de~Alfaro and T.A. Henzinger.
\newblock Interface theories for component-based design.
\newblock In {\em Proc. of EMSOFT}, LNCS 2211, pages 148--165. Springer, 2001.

\bibitem{WDHR06}
M.~De~Wulf, L.~Doyen, T.A. Henzinger, and J.-F. Raskin.
\newblock Antichains: A new algorithm for checking universality of finite
  automata.
\newblock In {\em Proc. of CAV}, LNCS 4144, pages 17--30. Springer, 2006.

\bibitem{DR10}
L.~Doyen and J.-F. Raskin.
\newblock Antichains algorithms for finite automata.
\newblock In {\em Proc. of TACAS}, LNCS 6015, pages 2--22. Springer-Verlag,
  2010.

\bibitem{DR11}
L.~Doyen and J.-F. Raskin.
\newblock Games with imperfect information: Theory and algorithms.
\newblock In {\em Lectures in Game Theory for Computer Scientists}, pages
  185--212. 2011.

\bibitem{EM79}
A.~Ehrenfeucht and J.~Mycielski.
\newblock Positional strategies for mean payoff games.
\newblock {\em Int. Journal of Game Theory}, 8(2):109--113, 1979.

\bibitem{EJ88}
E.A. Emerson and C.~Jutla.
\newblock The complexity of tree automata and logics of programs.
\newblock In {\em Proc. of FOCS}, pages 328--337. IEEE, 1988.

\bibitem{EJ91}
E.A. Emerson and C.~Jutla.
\newblock Tree automata, mu-calculus and determinacy.
\newblock In {\em Proc. of FOCS}, pages 368--377. IEEE, 1991.

\bibitem{fahrenberg_ICTAC11}
U.~Fahrenberg, L.~Juhl, K.G. Larsen, and J.~Srba.
\newblock Energy games in multiweighted automata.
\newblock In {\em Proc. of ICTAC}, LNCS 6916, pages 95--115. Springer, 2011.

\bibitem{WilkeBook}
E.~Gr{\"a}del, W.~Thomas, and T.~Wilke, editors.
\newblock {\em Automata, Logics, and Infinite Games: A Guide to Current
  Research}, LNCS 2500. Springer, 2002.

\bibitem{GH82}
Y.~Gurevich and L.~Harrington.
\newblock Trees, automata, and games.
\newblock In {\em Proc. of STOC}, pages 60--65. ACM, 1982.

\bibitem{FairSimulation}
T.A. Henzinger, O.~Kupferman, and S.~Rajamani.
\newblock Fair simulation.
\newblock {\em Information and Computation}, 173(1):64--81, 2002.

\bibitem{martin_AM75}
D.A. Martin.
\newblock Borel determinacy.
\newblock {\em Annals of Mathematics}, 102(2):363--371, 1975.

\bibitem{Mar98}
D.A. Martin.
\newblock The determinacy of {Blackwell} games.
\newblock {\em The Journal of Symbolic Logic}, 63(4):1565--1581, 1998.

\bibitem{pnueli_FOCS1977}
A.~Pnueli.
\newblock The temporal logic of programs.
\newblock In {\em Proc. of FOCS}, pages 46--57. IEEE Computer Society, 1977.

\bibitem{pnueli_POPL89}
A.~Pnueli and R.~Rosner.
\newblock On the synthesis of a reactive module.
\newblock In {\em Proc. of POPL}, pages 179--190, 1989.

\bibitem{rackoff_TCS78}
C.~Rackoff.
\newblock The covering and boundedness problems for vector addition systems.
\newblock {\em Theor. Comput. Sci.}, 6:223--231, 1978.

\bibitem{RamadgeWonham87}
P.J. Ramadge and W.M. Wonham.
\newblock Supervisory control of a class of discrete-event processes.
\newblock {\em SIAM Journal of Control and Optimization}, 25(1):206--230, 1987.

\bibitem{rosier_JCSS86}
L.E. Rosier and H.-C. Yen.
\newblock A multiparameter analysis of the boundedness problem for vector
  addition systems.
\newblock {\em J. Comput. Syst. Sci.}, 32(1):105--135, 1986.

\bibitem{solan}
E.~Solan.
\newblock Continuity of the value of competitive {M}arkov decision processes.
\newblock {\em Journal of Theoretical Probability}, 16(4):831--845, 2003.

\bibitem{Thomas97}
W.~Thomas.
\newblock Languages, automata, and logic.
\newblock In {\em Handbook of Formal Languages}, volume 3, Beyond Words,
  chapter~7, pages 389--455. Springer, 1997.

\bibitem{vardi_FOCS85}
M.Y. Vardi.
\newblock Automatic verification of probabilistic concurrent finite-state
  programs.
\newblock In {\em Proc. of FOCS}, pages 327--338. IEEE Computer Society, 1985.

\bibitem{VR11}
Y.~Velner and A.~Rabinovich.
\newblock Church synthesis problem for noisy input.
\newblock In {\em Proc. of FOSSACS}, LNCS 6604, pages 275--289. Springer, 2011.

\bibitem{Zie98}
W.~Zielonka.
\newblock Infinite games on finitely coloured graphs with applications to
  automata on infinite trees.
\newblock In {\em Theoretical Computer Science}, volume 200(1-2), pages
  135--183, 1998.

\bibitem{ZP95}
U.~Zwick and M.~Paterson.
\newblock The complexity of mean payoff games on graphs.
\newblock {\em Theoretical Computer Science}, 158:343--359, 1996.

\end{thebibliography}

\end{document}